\documentclass[conference]{IEEEtran}
\IEEEoverridecommandlockouts
\usepackage{cite}
\usepackage{amsthm}
\usepackage{etoolbox}
\usepackage{etoolbox}
\newtheorem{Theorem}{Theorem}
\newtheorem{remark}{Remark}
\usepackage{comment}
\usepackage{setspace}
\usepackage{caption}
\usepackage{xcolor}
\usepackage{epstopdf}
\usepackage{textcomp}
\usepackage{xcolor}    % To color text
\usepackage{colortbl}  % For table coloring
\usepackage{multirow}  % For multirow cells
\usepackage{tabularx}  % To automatically adjust column width
\usepackage{multirow}
\definecolor{myBlue}{RGB}{102, 140, 190}
\definecolor{myGreen}{RGB}{172, 196, 112}
\definecolor{myYellow}{RGB}{137, 133, 183}
\definecolor{myPurple}{RGB}{255, 220, 102}
\definecolor{orange}{rgb}{1.0, 0.647, 0.0}  % 自定义橙色
\usepackage{xcolor}
\definecolor{orange}{rgb}{1.0, 0.647, 0.0}  % 自定义橙色
\usepackage{xcolor}   % 用于颜色支持
\usepackage{multirow}  % 用于多行合并单元格
\usepackage{amsthm}
\usepackage{tabularx}
\usepackage{amsmath,amssymb,amsfonts}
\usepackage{algorithmic}
\usepackage{graphicx}
\usepackage{textcomp}
\usepackage{xcolor}
\usepackage{subfigure}
\usepackage{graphicx}
\usepackage{amsmath}
\usepackage{wrapfig}
\usepackage{amsthm,epstopdf}
\usepackage{float}
\usepackage{afterpage}
\newtheorem{myDef}{Definition}
\usepackage{setspace}
\usepackage{hyperref}
\usepackage{balance}    % 加载 balance 宏包，平衡双栏最后一页
\usepackage[ruled, vlined, linesnumbered]{algorithm2e}

\def\BibTeX{{\rm B\kern-.05em{\sc i\kern-.025em b}\kern-.08em
    T\kern-.1667em\lower.7ex\hbox{E}\kern-.125emX}}
\begin{document}

\title{Personalized 3D Spatiotemporal Trajectory Privacy Protection with Differential and Distortion Geo-Perturbation}
\author{Minghui Min,~\IEEEmembership{Member,~IEEE}, Yulu Li, Gang Li, Meng Li,~\IEEEmembership{Senior Member,~IEEE}, Hongliang Zhang,~\IEEEmembership{Member,~IEEE},
\\Miao Pan,~\IEEEmembership{Senior Member,~IEEE}, and Zhu Han,~\IEEEmembership{Fellow,~IEEE}

\thanks{ \emph{Corresponding author: Minghui Min}.}
\thanks{ Minghui Min, Yulu Li, and Gang Li are with the School of Information and Control Engineering, China University of Mining and Technology, Xuzhou 221116, China (email: minmh@cumt.edu.cn; liyulu@cumt.edu.cn; ligang@cumt.edu.cn).}
\thanks{Meng Li is with the Key Laboratory of Knowledge Engineering with Big Data, Ministry of Education, Hefei University of Technology, Hefei 230009, China, also with the School of Computer Science and Information Engineering, Hefei University of Technology, Hefei 230009, China, also with the Anhui Province Key Laboratory of Industry Safety and Emergency Technology, Hefei 230601, China, also with the Intelligent Interconnected Systems Laboratory of Anhui Province (Hefei University of Technology), Hefei 230009, China, and also with the Department of Mathematics and HIT Center, University of Padua, Padua 35122, Italy (email: mengli@hfut.edu.cn).}
\thanks{Hongliang Zhang is with the School of Electronics, Peking University, Beijing 100871, China (email: hongliang.zhang92@gmail.com).}
\thanks {Miao Pan is with the Department of Electrical and Computer Engineering, University of Houston, Houston, TX 77004, USA (email: miaopan.ufl@gmail.com).}
\thanks{ Zhu Han is with the Department of Electrical and Computer Engineering, University of Houston, Houston, TX 77004, USA, and also with the Department of Computer Science and Engineering, Kyung Hee University, Seoul 446-70, South Korea (email: hanzhu22@gmail.com).}
}
\maketitle

\begin{abstract}
 The rapid advancement of location-based services (LBSs) in three-dimensional (3D) domains, such as smart cities and intelligent transportation, has raised concerns over 3D spatiotemporal trajectory privacy protection. However, existing research has not fully addressed the risk of attackers exploiting the spatiotemporal correlation of 3D spatiotemporal trajectories and the impact of height information, both of which can potentially lead to significant privacy leakage.
\begin{comment}
However, the existing 2-dimensional (2D) trajectory privacy protection mechanisms cannot effectively defend against attacks by attackers who have knowledge of height dimension.
However, existing research has not fully considered the issue of attackers gaining access to the spatio-temporal correlation features of 3D trajectories, nor has it accounted for the height information of the trajectories, which could lead to significant privacy leakage risks. Furthermore, given that the sensitivity of different locations varies, users' privacy protection needs are not uniform.
\end{comment}
 To address these issues, this paper proposes a personalized 3D spatiotemporal trajectory privacy protection mechanism, named 3DSTPM. First, we analyze the characteristics of attackers that exploit spatiotemporal correlations between locations in a trajectory and present the attack model. Next, we exploit the complementary characteristics of 3D geo-indistinguishability (3D-GI) and distortion privacy to find a protection location set (PLS) that obscures the real location for all possible locations. To address the issue of privacy accumulation caused by continuous trajectory queries, we propose a Window-based Adaptive Privacy Budget Allocation (W-APBA), which dynamically allocates privacy budgets to all locations in the current PLS based on their predictability and sensitivity. Finally, we perturb the real location using the allocated privacy budget by the PF (Permute-and-Flip) mechanism, effectively balancing privacy protection and Quality of Service (QoS). Simulation results demonstrate that the proposed 3DSTPM effectively reduces QoS loss while meeting the user's personalized privacy protection needs.
\end{abstract}


\section{Introduction}
\begin{comment}
\IEEEPARstart{W}{ith} the continuous maturation of emerging technologies such as drone logistics\cite{lagorio2016research}, low-altitude monitoring\cite{7572034}, and urban air transportation, the application of LBSs in 3D space of fields such as smart cities and intelligent transportation is becoming increasingly widespread\cite{guo2022effects}. 3D spatio-temporal trajectory data provides strong data support for LBSs by accurately recording the location of objects in 3D space and their dynamic trajectories over time. However, the high precision and continuous nature of 3D spatio-temporal trajectory data pose significant privacy risks\cite{yan2024achieving}. If these data are improperly accessed or analyzed, attackers can infer sensitive information, including user's activity patterns, home and work addresses, and even gain insights into their lifestyle habits through trajectory analysis\cite{zhang2023personalized, wang2023efficient}.
\end{comment}
\IEEEPARstart{T}{he} application of spatiotemporal trajectory data in 3D spaces has become increasingly widespread, particularly in areas such as smart cities\cite{Mekdad2023},\cite{lagorio2016research}, traffic management, and low-altitude logistics\cite{guo2022effects},\cite{10.1109/TIV.2024.3483889}. Users perform continuous location-based queries on trajectories to obtain location-based services (LBS)\cite{9674220}. However, this querying process exposes a large amount of user data, including sensitive information such as health status, personally identifiable information, and other private details\cite{10793114},\cite{9646501}. If mishandled or exposed, it can pose significant risks to users' personal privacy\cite{jiang2021location}.

There is a significant spatiotemporal correlation between different locations in spatiotemporal trajectory. Attackers with knowledge of these correlations can leverage this information to deeply analyze users' behavior patterns and activity regularities, thereby making more accurate predictions of users' real trajectories\cite{10684714}. For example, in a smart city scenario, if a student takes an autonomous taxi to school at 8 a.m., an attacker with knowledge of spatiotemporal correlations could, by combining this with the user’s historical behavioral data, infer that the taxi is likely to head towards the school in the subsequent time period, rather than the nearby bank. Such inferences may lead to serious privacy leakage, potentially exposing the student's daily routines and personal habits. Furthermore, as users perform continuous location-based queries, more trajectory data is incrementally exposed, this allows attackers to accumulate additional information, further enhancing their predictive capabilities and significantly amplifying the associated privacy risks\cite{qiu2023novel}. Therefore, to effectively mitigate privacy leakage risks and protect users' trajectory privacy, it is essential to fully account for spatiotemporal correlations and limit the cumulative privacy leakage caused by continuous queries.
\begin{comment}
 By analyzing such correlations, attackers can more accurately predict the user's future movements, thus increasing the risk of privacy leakage.
 \end{comment}
\begin{comment}
Taking drone delivery services as an example, users regularly receive items from specific locations, such as hospitals and supermarkets, which are then delivered to their homes. During this process, the drone's flight trajectory data accurately records the starting location, destination, flight time, and stopover points of each flight. This data not only contains precise 3D geographical location information but also exhibits significant spatio-temporal correlation. Attackers can leverage this correlation to infer the user's behavior patterns by analyzing the 3D spatio-temporal trajectory data. For example, during multiple deliveries, the user's drone almost always departs from the hospital, first flies to the supermarket, and then to the user's home. This high-frequency and regular behavior pattern clearly reflects the user's activity sequence and habits, further revealing their lifestyle. By analyzing these 3D spatio-temporal data, attackers can not only infer the user's daily routines and lifestyle habits but may also expose sensitive personal information, such as health status and consumption habits, thereby severely violating the user's privacy.
\end{comment}

Moreover, compared to 2D trajectory data, 3D spatiotemporal trajectory data is larger in scale, higher in dimensionality, and contains richer semantic information. This complexity leads to significant variations in the sensitivity of different locations. In addition, the sensitivity of a given location may dynamically change over time due to contextual or behavioral dynamics, causing users’ privacy protection needs to vary even for the same location at different times\cite{10167666,10258267,10273430}. Applying a uniform protection level to all locations may result in insufficient protection for high-sensitivity locations and excessive protection for low-sensitivity locations\cite{LIU2024101074},\cite{10.1145/3474839}. Hence, it is essential to provide personalized privacy protection measures for different locations on 3D spatiotemporal trajectory based on users' specific privacy protection needs.

The PIM mechanism was proposed in \cite{xiao2015protecting} to protect trajectory privacy, considering the temporal correlation between locations in the trajectory. However, it does not fully account for the correlations between trajectory locations, overlooking spatial factors such as geospatial constraints, and applying a uniform protection level to all locations, which does not meet users' individualized privacy needs. To overcome these shortcomings, a spatiotemporal personalized trajectory privacy protection mechanism PTPPM was proposed in \cite{cao2024protecting}, which more comprehensively considers the spatiotemporal correlations between trajectory locations and provides varying levels of protection based on users’ specific privacy requirements. Nevertheless, PTPPM overlooks the issue of cumulative privacy leakage caused by continuous queries along the trajectory, nor does it account for the potential variation in location sensitivity over time. Furthermore, the above mechanism is designed for 2D space and it is difficult to handle 3D spatiotemporal trajectory data that incorporates height information. Once attackers gain access to the height information of the trajectory, the risk of privacy leakage increases\cite{sym16091248, LI2021102323, zheng2022semantic}.

In 3D space, existing related mechanisms mainly focus on protecting individual location points, overlooking the spatiotemporal correlations between locations in a trajectory. Inspired by the concept of differential privacy\cite{dwork2006differential}, 3D geo-indistinguishability (3D-GI) was developed in \cite{9646489}, providing a rigorous method to quantify location privacy in 3D space, laying the foundation for 3D location privacy protection mechanisms. Building upon this, the complementary characteristics of 3D-GI and distorted privacy were analyzed, and a personalized location privacy protection mechanism, P3DLPPM, was proposed in \cite{10325612}, which enhances the robustness of the mechanism while meeting users' individualized needs. However, this mechanism is ineffective against attackers who exploit spatiotemporal correlations, directly applying it to trajectory privacy protection could lead to significant privacy leakage. Therefore, providing personalized privacy protection for 3D spatiotemporal trajectory remains an urgent challenge that needs to be addressed.

In this paper, we propose a personalized 3D spatiotemporal trajectory privacy protection mechanism, named 3DSTPM, to protect the user's trajectory privacy in 3D space. To achieve this, we generate a transfer probability matrix based on the user's historical behavior and data, quantifying the spatiotemporal correlation between trajectory locations, thereby resisting attackers who can analyze spatiotemporal characteristics. Furthermore, we combine the complementary properties of 3D-GI\cite{9646489} and distortion privacy (i.e., expected inference error)\cite{yu2017dynamic}. 3D-GI can limit the posterior information obtained by an attacker, but cannot completely defend against an attacker with prior knowledge, while distortion privacy ensures that the attacker's inference error remains above the threshold, but does not limit the exposure of a posteriori information. By combining the two, we improve the robustness of our mechanisms against attackers with knowledge of spatiotemporal information, while meeting the individual needs of our users. In addition, we dynamically adjust the privacy budget according to the user's personalized needs, and provide corresponding privacy protection for locations with different sensitivities, enabling more flexible privacy protection tailored to user-specific contexts.

The sliding window mechanism can effectively group continuous queries within a specified time interval\cite{qiu2023novel}, thereby limiting cumulative privacy leakage. Building upon this, we introduce the Window-based Adaptive Privacy Budget Allocation (W-APBA) algorithm to mitigate privacy leakage from continuous queries in 3D space. Specifically, the algorithm dynamically allocates privacy budgets to different locations based on their sensitivity at various times, meeting users' fluctuating privacy protection needs. Additionally, the sliding window mechanism constrains the total privacy budget within each window, further reducing cumulative leakage from consecutive queries along the trajectory.
\begin{figure*}[t]
	\centerline{\includegraphics[width=12.5cm]{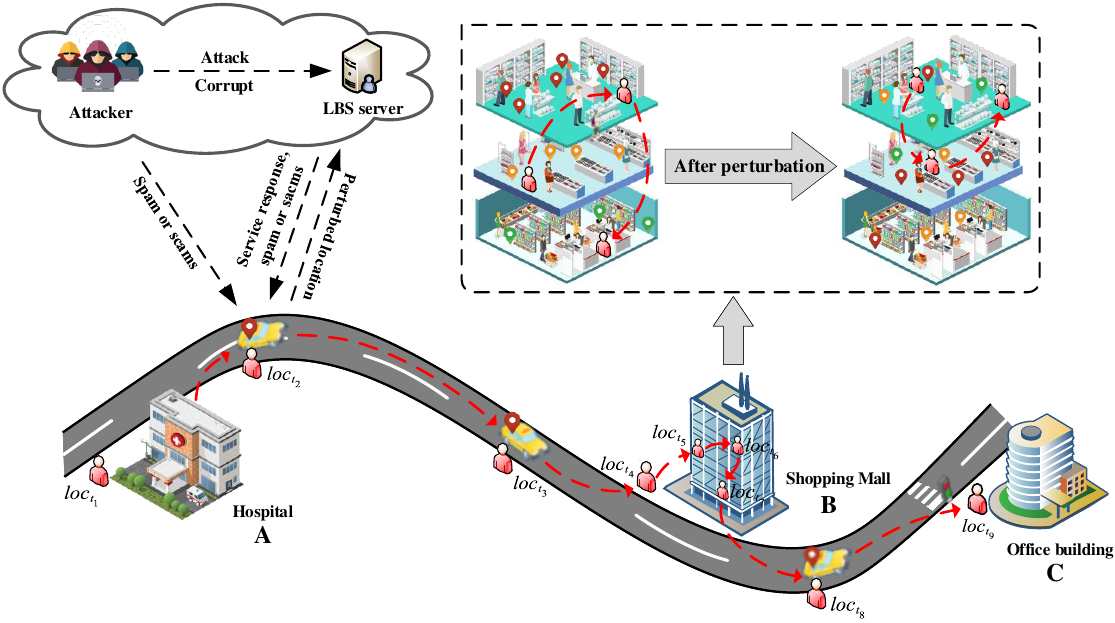}}
	\caption{Illustration of 3D spatiotemporal trajectory privacy protection, the user moves between three large buildings and requests LBS services in both indoor and outdoor scenarios. To prevent attackers from obtaining the user’s personal information from the obtained trajectory data, the user uploads perturbed locations to the LBS server to obtain the corresponding service.}
	\label{figl}
\vspace{-0.2cm}
\end{figure*}

Furthermore, building upon the PF (Permute-and-Flip) mechanism proposed in our previous work \cite{10325612}, we perturb the real locations at different timestamps along the trajectory. Additionally, we theoretically demonstrate that the distance between the real and perturbed locations generated by the PF mechanism is constrained by a controllable upper bound. This enhancement further strengthens the theoretical foundation of the PF mechanism in balancing user privacy protection requirements with QoS. Simulation results demonstrate that 3DSTPM provides personalized trajectory privacy protection, and offers superior privacy protection compared to P3DLPPM\cite{10325612}, 2DPTPPM\cite{cao2024protecting} and 3DPIM. The main contributions of our work include:
\begin{comment}
\begin{itemize}
\item We design a 3D trajectory privacy protection mechanism that accounts for spatio-temporal correlations. By incorporating the height information of locations on the trajectory, the mechanism protects 3D trajectory privacy, overcoming the limitations of traditional planar trajectory privacy protection. Additionally, it fully considers the spatio-temporal correlations between locations, effectively defending against attackers who exploit such correlations.
\item We propose a $w$-window based privacy budget allocation algorithm, taking into account the location sensitivity of the user, to dynamically allocate privacy budgets for different locations on the trajectory, which meets the user's personalized needs and avoids the accumulation of privacy budgets due to the continuous querying on the trajectory. The allocated privacy budgets satisfy the $w$-trajectory differential privacy, which guarantees the user's privacy protection needs.
\item We conduct simulations to study the effects of privacy parameters $\epsilon$ and $Em$ on user trajectory privacy and Qos loss, and we verify that under the same $\epsilon$ and $Em$, the trajectory privacy of 3DPTPPM is much larger than that of 2DPTPPM, and the comprehensive performance is better than that of P3DLPP.
\end{itemize}
\end{comment}
\begin{itemize}
\item We propose a personalized 3D spatiotemporal trajectory privacy protection mechanism, named 3DSTPM, to protect users' trajectory privacy in 3D space. This mechanism effectively tackles the challenges posed by the high dimensionality of 3D space, expanding its applicability. Besides, it considers the spatiotemporal correlations between locations on the trajectory, and integrates the complementary characteristics of 3D-GI and distorted privacy, enhancing the overall robustness of the protection.
\item We propose a W-APBA algorithm to meet the user's dynamically changing privacy protection needs. Specifically, the algorithm adapts to temporal variations in location sensitivity, thereby enabling more flexible privacy protection. It also incorporates a sliding window mechanism to limit the cumulative privacy leakage caused by continuous queries, ensuring robust privacy protection. Based on the allocated privacy budget, the perturbed location is generated using the PF mechanism, effectively balancing privacy protection and QoS.
\item We conduct simulations to study the effects of key privacy parameters, including privacy budget, expected inference error bound, and window size, on trajectory privacy and QoS loss. In addition, we compare and analyze the overall performance of 3DSTPM against benchmarks to demonstrate the advantage of our proposed mechanism.
\end{itemize}

The remainder of this paper is organized as follows. Section \ref{Related Work} reviews the related work. Section \ref{SYSTEM MODEL} presents the system model and we present the trajectory privacy protection statement in Section \ref{LOCATION PRIVACY NOTIONS AND TRAJECTORY FEATURES}. A 3DSTPM framework is proposed in Section \ref{PERSONALIZED TRAJECTORY PRIVACY PROTECTION MECHANISM}. The evaluation results are provided in Section \ref{SIMULATION RESULTS}. Finally, we conclude this work in Section \ref{CONCLUSION}.
\section{Related Work}\label{Related Work}
Existing research on privacy protection in 3D space primarily focuses on individual location protection\cite{10766641,qiu2020location,hu2023federated,10815979}, such as anonymity \cite{4359010}, location perturbation \cite{1250908}, and generalization\cite{glove2020privacy}. A mechanism based on anonymity, 3D Clique Cloak, was proposed in \cite{9705498}, to safeguard users' location privacy in 3D space. Nonetheless, such $k$-anonymity-based privacy protection mechanisms\cite{di2023k} rely heavily on trusted third parties, which introduce potential privacy leakage in the event of server failures or malicious attacks. To address this limitation, a 3D spatial personalized location privacy protection mechanism based on differential privacy, P3DLPPM, was proposed in \cite{10325612}, integrating the complementary characteristics of 3D-GI and distortion privacy to enhance the mechanism's robustness. However, these methods were designed specifically for location privacy protection and overlook spatiotemporal correlations between different locations. When directly applied to trajectory privacy protection, they cannot defend against attackers who exploit such correlations, increasing the risk of privacy leakage.

Spatiotemporal trajectory exhibits spatiotemporal correlations between locations, possibly exposing users' behavioral patterns. Once attackers exploit these correlations, they can accurately infer sensitive personal information \cite{10684714},\cite{shen2024bigru}. A PIM mechanism was proposed in \cite{xiao2015protecting} to protect trajectory privacy, considering the temporal correlations between locations on the trajectory. It hides the real location within a set of possible locations filtered based on prior knowledge, thereby reducing the accuracy of attackers' inferences. Nonetheless, this mechanism overlooks the spatial correlations between locations on the trajectory. On this basis, the concept of spatiotemporal event privacy was introduced in \cite{8946543}, which considers both temporal and spatial correlations between locations on the trajectory simultaneously and provides a quantifiable privacy metric. However, the above mechanisms do not consider the height information of the trajectory and are only suitable for 2D space. If directly applied to 3D spatiotemporal trajectory privacy protection, they will be unable to effectively defend against attackers who possess height information, leading to significant privacy leakage risks.

The sensitivity of different locations on a trajectory varies, and even the sensitivity of the same location may change over time. Applying uniform protection across all locations on the trajectory can not address personalized privacy protection needs of users\cite{cao2024protecting},\cite{qiu2020mobile}. By incorporating the temporal correlation of locations,  personalized privacy protection for different locations on the trajectory based on user needs is provided in \cite{cao2024protecting}. However, this mechanism only analyzes the impact of privacy parameters on the degree of personalized privacy protection using simplified models, without fully accounting for the factors that influence trajectory sensitivity. To address its limitations, a false location generation mechanism based on a semantic tree was proposed in \cite{qiu2020mobile} for personalized privacy protection, comprehensively considering the factors influencing location sensitivity. However, this mechanism does not account for the fact that the sensitivity of different locations on the trajectory may change over time, making it unable to adapt to the evolving privacy protection needs of users.
\section{System Model}\label{SYSTEM MODEL}
\begin{comment}
We consider the trajectory of objects with locationg function devices in 3D space, such as UAV performing logistics and distribution, users requesting LBS indoors, mobile sources in underwater sensor networks, etc. They share location information with the LBS server to obtain LBS. When their location is in the 2D plane, the altitude information can be considered to be 0.
Take the example of a mobile user carrying a mobile phone with locationg function and requesting LBS in a large building. The user share their location information with an LBS server for networked ride-hailing, indoor navigation, and mobile attendance. Figure 1 illustrates the user's mobile trajectory, where A,B,C are three large buildings containing multiple floors, and the user takes a cab to the next building, $loc_{t_1},loc_{t_2},loc_{t_3}...loc_{t_8}$ indicating the user's location at that moment. To protect privacy, the user releases the perturbed location. An untrusted LBS server can also be compromised by an external attacker, which can analyze the received location information to infer the user's actual location at that particular time and send related spam or scams while providing feedback services.
\end{comment}
\begin{figure}[t]
	\centerline{\includegraphics[width=7.5cm]{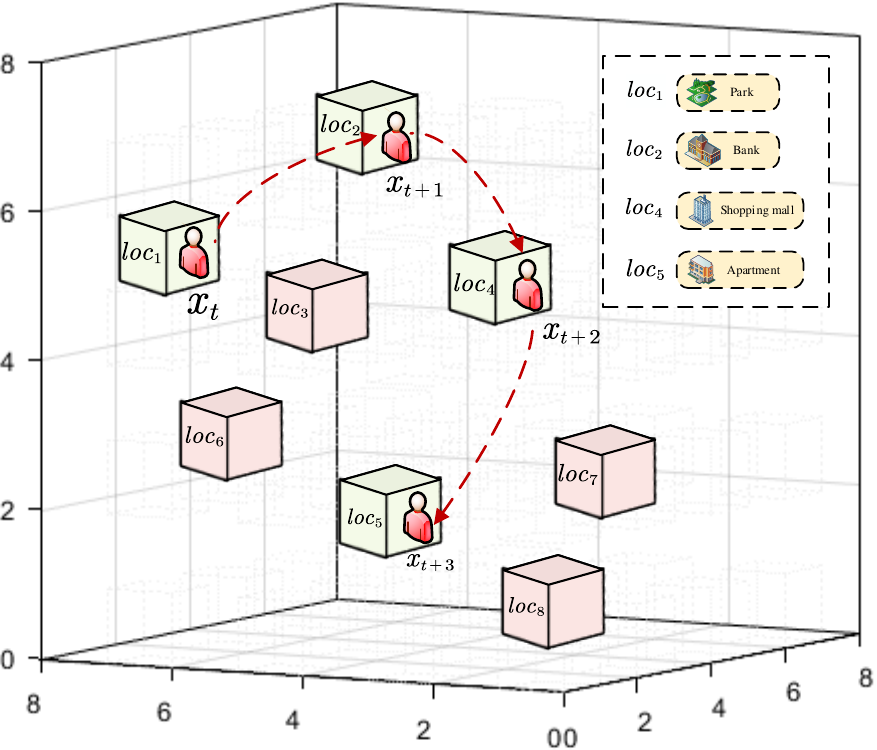}}
	\caption{User 3D spatiotemporal trajectory coordinate mapping and coordinate status.}
	\label{fig2}
\end{figure}

We consider an LBS system that consists of mobile users equipped with location enabled devices in 3D space (including both outdoor scenarios and indoor buildings), and an untrusted LBS server,  as shown in Fig. \ref{figl}. For outdoor scenarios, users share their location data with the LBS server to obtain services, such as vehicle users' navigation for live traffic information. For indoor scenarios, users can access LBS in high-rise indoor buildings for indoor navigation. Fig. \ref{figl} depicts the user's movement trajectory, with points A, B, and C representing three large buildings. The differently colored markers represent the potential locations of users across various floors in a building. Each floor offers distinct services and experiences tailored to the user's needs, similar to that in \cite{10325612}. The user, travelling by vehicle from one building to another, generates trajectory points denoted as $loc_{t_1},loc_{t_2},\ldots,loc_{t_9}$. Attackers may corrupt untrusted LBS servers to infer the user’s personal information from the obtained trajectory data, risking spam or fraud. To protect users' trajectory privacy, the user uploads perturbed locations, thereby reducing the risk of exposure to their actual trajectory. Important symbols are summarized in Table \ref{tab:2}.
\subsection{User Model}
\begin{comment}
We consider trajectory of the user involves three large buildings (a hospital, a shopping center, and an office building) as well as roads between the buildings. Different floors of each building provide different services and experiences for users. In a shopping center, for example, the first floor is dedicated to daily necessities, the second floor to cosmetics, and the third floor to medicinal herbs. Users go to different floors of different buildings according to their needs.
\end{comment}
\begin{comment}
We divided the space involved in the trajectory, which contain all possible locations of the user's into subspaces, each of which represents the user's location status and each of which is associated with unique 3D coordinates. And when the user is located on a road between buildings or on the first floor of a building, we consider the user's location to be a z-axis coordinate of 0.
\end{comment}
\begin{comment}
To achieve a more precise and concise representation of user locations and location states,
\end{comment}
We represent the user's trajectory in 3D space as $Tr=<\boldsymbol{x}_1,\boldsymbol{x}_2,\ldots,\boldsymbol{x}_t,\ldots,\boldsymbol{x}_T>$, where $\boldsymbol{x}_t$ denotes the user's 3D location at time $t$, and $T$ indicates the total length of the trajectory. The space encompassing all locations along the user's trajectory is divided into multiple subspaces, each corresponding to a unique 3D coordinate that represents a possible user location. The set of all possible locations of the user within the region is denoted as $\mathcal{A}$, where $
\mathcal{A} =\{loc_1,loc_2,\ldots ,loc_N\}
$, $N$ represents the total number of locations.
\begin{comment}
The user's real location at time $t$ is given by $\boldsymbol{x}_t$, while $l_t$ represents the 3D coordinates of the user’s location state at time $t$.
\end{comment}
As shown in Fig. \ref{fig2}, $
\mathcal{A} =\{loc_1,loc_2,\ldots ,loc_{9}\}
$, and the user's location at time $t$ is given by $\boldsymbol{x}_t=loc_1=[2,6,5]$.

To protect privacy, the user applies a location perturbation mechanism, which transforms the actual location $\boldsymbol{x}_t$ from the original location set $O_1$ into a fake location $\boldsymbol{x}_{t}^{'}$ drawn from the perturbed location set $O_2$. The probability distribution for location perturbation, denoted as $f$, given by
\begin{equation}
	f\left( \boldsymbol{x}_{t}^{'}|\boldsymbol{x}_t \right) =\text{Pr}\left( O_2=\boldsymbol{x}_{t}^{'}|O_1=\boldsymbol{x}_t \right) ,\ \ \ \ \boldsymbol{x}_t,\boldsymbol{x}_{t}^{'}\in \mathcal{A}.
\end{equation}
\begin{table}[!t]%调节图片位置，h：浮动；t：顶部；b:底部；p：当前位置
	\centering
	\caption{List of Notations.}
	\label{tab:2}
	\begin{tabular}{ll}%表格中的数据居中，c的个数为表格的列数
		\hline\hline\noalign{\smallskip}	
		Symbol & Description  \\
		\noalign{\smallskip}\hline\noalign{\smallskip}
        $\mathcal{A}$ & User map\\
        $D\left( \Phi _t \right)$ & Diameter of $\Phi _t$\\
        $E_m$ & Expected inference error bound\\
        $loc_i$ & Location $i$\\
        $LS_{i,t}$ & Location sensitivity of $loc_i$ at time $t$\\
        $LP_{i,t}$ & Location predictability of $loc_i$ at time $t$\\
        $\mathbf{M}$ & Location transition probability matrix\\
        $\mathbf{p}_{t}^{-}/\mathbf{p}_{t}^{+}$ & prior/posterior probability at time $t$\\
        $w$ & Size of sliding window\\
        $\boldsymbol{x}_t/\boldsymbol{x}_{t}^{'}/\boldsymbol{\hat{x}}_t$ & Real/perturbed/inferred location at time $t$\\
        $\epsilon _w$ & Total privacy budget of sliding window\\
        $\delta$ & Probability thresholds, $0<\delta <1$\\
        $\varDelta \chi _t$ & Possible location set at time $t$\\
        $\Phi _t$ & Protection location set (PLS) at time $t$\\
		\noalign{\smallskip}\hline
	\end{tabular}
\end{table}

We define $\mathbf{p}_t$ as the user’s location state at time $t$, and $\mathbf{p}_t\left[ i \right] =\text{Pr}\left( \boldsymbol{x}_t=loc_i \right)$, where $\boldsymbol{x}_t$ is the real location of the user at time $t$, $loc_i$ is the $i_{th}$ location of all possible locations of the user, $\mathbf{p}_t\left[ i \right] $ is the probability that the user is located $loc_i$ at time $t$. Assuming that the user has an equal probability of being at any given location, the set of locations $\mathcal{A}=\left\{ loc_1,\ loc_2,\ loc_6,\ loc_8 \right\}$, and the probability distribution of user locations is given by $\mathbf{p}_t=\left[ 0.25,0.25,0,0,0,0.25,0,0.25,\cdot \cdot \cdot ,0 \right]$. The prior and posterior probabilities of the user, denoted as $\mathbf{p}_{t}^{-}$ and $\mathbf{p}_{t}^{+}$, represent the probabilities before and after observing the released perturbed location $\boldsymbol{x}_{t}^{'}$, respectively.

\subsection{Attack Model}
\begin{comment}
We consider the attackers to be untrusted LBS servers or external attackers that may attack or destroy LBS servers. In addition to the priori probability distribution $\mathbf{p}_{t}^{-}$ of the user's location and the probability distribution $f\left( \boldsymbol{x}_{t}^{\prime}|\boldsymbol{x}_t \right)$
of the perturbed location mastered by the traditional attackers, they can also obtain the user's location transfer probability matrix $\mathbf{M}
$ based on the user's historical trajectory data and behavioral habits, and increase the success rate of inferring the user's real trajectory based on the spatio-temporal correlation between different locations on the trajectory, and the attacker's specific process of inferring the user's real trajectory is as follows:
\end{comment}
Attackers can be either untrusted LBS servers or external entities attempting to compromise or attack local service system servers. We assume that attackers possess knowledge of spatiotemporal correlations\cite{cao2024protecting}. In addition to the prior probability distribution of user locations $\mathbf{p}_{t}^{-}$, and the probability distribution of perturbed locations $f\left( \boldsymbol{x}_{t}^{\prime}|\boldsymbol{x}_t \right)$, which are typically known to traditional attackers, they can also derive the user's location transition probability matrix $\mathbf{M}$ based on their historical behavior \cite{xiao2015protecting}. By leveraging this additional information, attackers can significantly enhance their ability to infer the user's real trajectory. The specific process by which attackers infer the user's real trajectory is as follows:

After observing the perturbed location posted by the user, the attacker can compute the posterior probability distribution based on the existing prior knowledge, i.e:
\begin{equation}\label{eq:3}
\mathbf{p}_{t}^{+}=\mathrm{Pr}\left( \boldsymbol{x}_t|\boldsymbol{x}_{t}^{\prime} \right) =\frac{\mathrm{Pr}\left( \boldsymbol{x}_t \right) f\left( \boldsymbol{x}_{t}^{\prime}|\boldsymbol{x}_t \right)}{\sum_{\boldsymbol{x}_t\in \mathcal{A}}{\mathrm{Pr}}\left( \boldsymbol{x}_t \right) f\left( \boldsymbol{x}_{t}^{\prime}|\boldsymbol{x}_t \right)}.
\end{equation}

A Bayesian attacker can employ an optimal inference attack, i.e., inferring the real location by minimizing the expected inference error of the posterior distribution that depends on the perturbed location\cite{yu2017dynamic}. Thus, the inferred location can be expressed as
\begin{equation}
\hat{\boldsymbol{x}}_t=\mathrm{arg}\min_{\hat{\boldsymbol{x}}_t\in \mathcal{A}} \sum_{\boldsymbol{x}_t\in \mathcal{A}}{\mathrm{Pr}}\left( \boldsymbol{x}_t|\boldsymbol{x}_{t}^{^{\prime}} \right) d_3\left( \hat{\boldsymbol{x}}_t,\boldsymbol{x}_t \right) ,
\end{equation}
where $d_3\left( \hat{\boldsymbol{x}}_t,\boldsymbol{x}_t \right) $ denotes the Euclidean distance between the attacker's inferred location and the user's real location in 3D space at time $t$.

The prior knowledge of an attacker who has knowledge of the spatiotemporal correlation between different locations on the user's trajectory is not static, and it can be constantly updated with the prior distribution probability of the user's location at the next time as
 \begin{equation}\label{eq:4}
\mathbf{p}_{t+1}^{-}=\mathbf{p}_{t}^{+}\mathbf{M}.
\end{equation}

The attacker updates the prior probability to compute the posterior probability for the next time step and derives the inferred location at time $t+1$ according to (\ref{eq:4}). By iteratively performing optimal inference on the user's location at each time step, the attacker can reconstruct the user's entire trajectory, thereby compromising their trajectory information.
\section{Trajectory Privacy Notions and Problem Statement}\label{LOCATION PRIVACY NOTIONS AND TRAJECTORY FEATURES}
In this section, we first list the trajectory features and mainly used trajectory privacy notions. Then, we present this paper’s problem statement.
\subsection{3D Spatiotemporal Trajectory}
The 3D spatiotemporal trajectory extends the traditional 2D trajectory model by incorporating additional height dimension, such as altitude or floor level, to more precisely represent a user's movement patterns over both space and time. Here, we provide the definition of a 3D spatiotemporal trajectory.
\begin{myDef}
\textnormal{\textbf{(3D Spatiotemporal Trajectory)}} A 3D spatiotemporal trajectory is defined as an ordered sequence of points that represent the movement of an object within a specific time interval in 3D space. It can be mathematically expressed as
\begin{equation}
Tr=<\boldsymbol{x}_1,\boldsymbol{x}_2,\ldots,\boldsymbol{x}_t,\ldots,\boldsymbol{x}_T>,
\end{equation}
 where $\boldsymbol{x}_t$ denotes the user's 3D location at time $t$, and $T$ represents the total length of trajectory. Consequently, the trajectory $Tr$ can also be formulated as
 \begin{equation}
 Tr=\{\boldsymbol{x}_t\mid t\in \{1,2,\dots ,T\}\}.
 \end{equation}
\end{myDef}
The locations at adjacent moments in the trajectory are spatially and temporally correlated\cite{liu2017spatiotemporal,qiu2023novel}. Specifically, the likelihood of a user appearing at a particular location at time $t$ is influenced by the location at the previous time step, and, in turn, affects the probability of the user being at a specific location at $t+1$. Moreover, two locations that are not geographically accessible to each other cannot appear consecutively at adjacent time steps. This limitation arises from the inherent spatial constraints, which restrict the feasibility of transitions between locations.

\subsection{Location Transition Probability Matrix}
The location transition probability matrix $\mathbf{M}$ represents the probability of a user transferring between two locations, which we assume that it remains constant over time. The element in the ith row and jth column of matrix $\mathbf{M}$ is $m_{ij}$, and $m_{ij}=\frac{n_{ij}}{\sum_j{n_{ij}}}$, where $m_{ij}$ denotes the probability of a user transferring from $loc_i$ to $loc_i$, $n_{ij}$ is an element in the location transfer matrix $\mathbf{N}$, which denotes the number of times the user has traveled from $loc_i$ to $loc_j$.

The matrix $\mathbf{M}$ can reflect the spatiotemporal correlation between the locations on the trajectory, and the prior probability distribution $\mathbf{p}_{t}^{-}$ of the user at time $t$ can be computed from the posterior probability distribution $\mathbf{p}_{t-1}^{+}$ at time $t-1$ and the matrix $\mathbf{M}$. Moreover, when the two locations are not geographically spatially reachable $m_{ij}=0$.

\subsection{Differential Geo-Perturbation in 3D Space}\label{AA}
3D geo-indistinguishability (3D-GI) is a privacy concept proposed in\cite{9646489}, which extends geo-indistinguishability to 3D space. It ensures that the perturbed locations generated for any two locations within a sphere of radius $r$ have similar probability distributions, making it geographically indistinguishable to an attacker for all locations within it. In the sequel, we define the sphere as a protected area, and all locations within it constitute protected location sets (PLS). Expanding on this notion, we enhance the initial 3D-GI, originally tailored for singular locations, to encompass trajectories spanning multiple timestamps. Next, we present a refined definition of 3D-GI at each timestamp along the trajectory.
\begin{comment}
\begin{myDef}
\textnormal{\textbf{(3D-Geo-Indistinguishability(3D-GI))}} A mechanism that makes all locations $x$ and $y$ in the 3D space satisfy $
\epsilon _g$-geo-differential privacy, given by
\begin{equation}
\frac{f\left(x^{\prime}|x\right)}{f\left(x^{\prime}|y\right)}\leq e^{\epsilon_gd_3(x,y)},
\end{equation}
where $d_3(x,y)$ is the Euclidean distance between $x$ and $y$ in 3D space, $\epsilon_{g}$ is determined by privacy budgets $\epsilon$ and user-centered spheres D, and $\epsilon _g=\small{\frac{\epsilon}{D}}$.
\end{myDef}
\end{comment}
\begin{myDef}
\textnormal{\textbf{(3D Geo-Indistinguishability, 3D-GI)}} For a user’s possible location set $D_t \in \mathbb{R}^3$ at time $t$, and any given locations $\boldsymbol{x}_t$ and $\boldsymbol{y}_t$, where $t\in [1,2,\ldots ,T]$, the perturbation mechanism $\mathcal{M}_t$ satisfies 3D-GI if and only if the following inequality holds:
\begin{equation}\label{eq:7}
\mathrm{Pr[}\mathcal{M} _t(\boldsymbol{x}_{t}^{\prime}|\boldsymbol{x}_t)]\le e^{\epsilon d_3(\boldsymbol{x}_t,\boldsymbol{y}_t)}\mathrm{Pr[}\mathcal{M} _t(\boldsymbol{x}_{t}^{\prime}|\boldsymbol{y}_t)].
\end{equation}
\end{myDef}
\begin{comment}
3D-GI ensures that fake locations generated at two 3D geographical locations close to each other have similar probability distributions. Privacy budgets and user-centered spheres determine the parameter $\epsilon_{g}$.
\end{comment}
\begin{comment}
All locations within the sphere have similar probability distributions and are geographically indistinguishable from the attacker. We define this sphere as the protection area, and all the locations within it are the PLS. In general, let $\epsilon=\epsilon_{g}\cdot D$, and D is the protection spherical area’s diameter.
\end{comment}
\begin{comment}
For any two locations $x$ and $y$ in the PLS satisfy $\epsilon$-geo-indistinguishability, we have
\begin{equation}\label{eq:6}
e^{-\epsilon}\leq\frac{f\left(x^{\prime}|x\right)}{f\left(x^{\prime}|y\right)}\leq e^\epsilon.
\end{equation}
\end{comment}
By defining $ \epsilon_t = \epsilon d_3(\boldsymbol{x}_t, \boldsymbol{y}_t) $, where $d_3(\boldsymbol{x}_t, \boldsymbol{y}_t)$ denotes the Euclidean distance between $\boldsymbol{x}_t$ and $\boldsymbol{y}_t$ in 3D space, we establish that the mechanism $\mathcal{M}_t$ satisfies $ \epsilon_t $-differential privacy, as stated in (\ref{eq:7}). This ensures that the probability distributions of the perturbed locations corresponding to the real locations $\boldsymbol{x}_t$ and $\boldsymbol{y}_t$ remain similar, thereby limiting an attacker's ability to accurately infer the user's real location based solely on the observed perturbed data. The degree of similarity between the perturbed and real locations is governed by the parameter $ \epsilon_t $, where a smaller $ \epsilon_t $ value indicates a higher level of privacy protection for the user's location information.

According to (\ref{eq:7}) we derive
\begin{equation}
\mathrm{Pr(}\boldsymbol{x}_t|\boldsymbol{x}_{t}^{\prime})\le e^{\epsilon _t}\frac{\mathrm{Pr(}\boldsymbol{x}_t)}{\sum_{\boldsymbol{y}_t\in D_t}{\mathrm{Pr(}\boldsymbol{y}_t)}}.
\end{equation}

This result indicates that, regardless of the prior knowledge an attacker possesses, 3D-GI bounds the multiplicative distance between posterior and prior probabilities within $ \epsilon_t $, effectively limiting the attacker's ability to infer posterior information.
\begin{comment}
\textbf{Upper bound of posterior probability:} For a PLS $\Phi _t$ corresponding to any possible location $x_t\in A$
  of the user, we obtain an upper bound on the posterior probability distribution of any observed false location by the following derivation:
 \begin{footnotesize}
\begin{equation}\label{eq:1}
\begin{aligned}
	\mathrm{Pr(}x_t|{x_t}^{\prime})&=\frac{\mathrm{Pr(}x_t)f(x_t\prime|x_t)}{\sum_{y_t\in A}{\mathrm{Pr(}y_t)}f(x_t\prime|y_t)}\\
	&=\frac{\mathrm{Pr(}x_t)f(x_t\prime|x_t)}{\sum_{y_t\in \Phi _t}{\mathrm{Pr(}y_t)}f(x_t\prime|y_t)+\sum_{y\in A\setminus \Phi _t}{\mathrm{Pr(}y_t)}f(x_t\prime|y_t)}\\
	&\le \frac{\mathrm{Pr(}x_t)f(x_t\prime|x_t)}{\sum_{y_t\in \Phi _t}{\mathrm{Pr(}y_t)}f(x_t\prime|y_t)}\\
	&=\frac{\mathrm{Pr(}x_t)}{\sum_{y_t\in \Phi _t}{\mathrm{Pr(}y_t)}f(x_t\prime|y_t)/f(x_t\prime|x_t)}\\
	&\le e^{\epsilon}\frac{\mathrm{Pr(}x_t)}{\sum_{y_t\in \Phi _t}{\mathrm{Pr(}y_t)}},\\
\end{aligned}
\end{equation}
\end{footnotesize}
\end{comment}
\subsection{Distortion Geo-Perturbation (Expected Inference Error)}\label{SCM}
\begin{comment}
\begin{myDef}
\textnormal{\textbf{(Expected Inference Error)}} Expected error can be categorized into conditional expected error and unconditional expected error. The conditional expected error is the expected error of an attacker with priori knowledge to infer the real location of the user at time $t$, while the unconditional expected error is the overall average error of an attacker without a priori knowledge to infer the location of the user at the at time $t$ in all cases.
\end{myDef}
\end{comment}
The expected inference error\cite{yu2017dynamic} is one of the key methods for quantifying the effectiveness of trajectory privacy protection. It measures the effectiveness of privacy protection by evaluating the level of error an attacker might incur when attempting to infer the user's real location. The definition is described as follows:
\begin{myDef}
\textnormal{\textbf{(Expected Inference Error)}} An expected inference error can be categorized into the conditional expected inference error and unconditional expected inference error. The conditional expected inference error refers to the expected error incurred by an attacker with prior knowledge attempting to infer the real location of the user at time $t$. In contrast, the unconditional inference expected error represents the average error of an attacker who lacks prior knowledge and attempts to infer the user's location at time $t$ in all scenarios.
\end{myDef}
\begin{comment}
3D-GI can limit the posteriori information obtained by an attacker, but it cannot quantify the user's location privacy. in previous work, information entropy and k-anonymity have been used for privacy metrics, but the accuracy is not high. To improve the accuracy, literature [] proposes to use expected inference error to quantify user location privacy.
\end{comment}
We assume that the attacker who has the knowledge of spatiotemporal correlation can observe all perturbed locations posted at previous moments. Under such conditions, the conditional expected inference error is defined as
\begin{equation}\label{eq:9}
ExpEr\left( \boldsymbol{x}_{t}^{^{\prime}} \right) =\min_{\hat{\boldsymbol{x}}_t\in \mathcal{A}} \sum_{\boldsymbol{x}_t\in \mathcal{A}}{\mathrm{Pr}\left( \boldsymbol{x}_t|\boldsymbol{x}_{t}^{^{\prime}} \right) d_3\left( \hat{\boldsymbol{x}}_t,\boldsymbol{x}_t \right) .}
\end{equation}
\begin{comment}
\textbf{Lower bound of inference error:}
\end{comment}

When the attacker possesses sufficient prior knowledge, their inferred location is constrained to include the user’s real location within the PLS, minimizing the inference error and maximizing the user's privacy protection level\cite{10325612}. The inference error is then given by
\begin{equation}
\begin{split}
&\min_{\hat{\boldsymbol{x}}_t\in \mathcal{A}} \sum_{\boldsymbol{x}_t\in \Phi _t}{\frac{\mathrm{Pr}\left( \boldsymbol{x}_t|\boldsymbol{x}_{t}^{\prime} \right)}{\sum_{\boldsymbol{y}_t\in \Phi _t}{\mathrm{Pr}\left( \boldsymbol{y}_t|\boldsymbol{x}_{t}^{\prime} \right)}}d\left( \hat{\boldsymbol{x}}_t,\boldsymbol{x}_t \right)}\\
	&\ge e^{-\epsilon}\min_{\hat{\boldsymbol{x}}_t\in \Phi _t} \sum_{\boldsymbol{x}_t\in \Phi _t}{\frac{\mathrm{Pr}\left( \boldsymbol{x}_t \right)}{\sum_{\boldsymbol{y}_t\in \Phi _t}{\mathrm{Pr}\left( \boldsymbol{y}_t \right)}}d\left( \hat{\boldsymbol{x}}_t,\boldsymbol{x}_t \right).}\\
\end{split}
\end{equation}

\subsection{$\delta$-Location Set}
To enhance the protection of locations frequently visited by users, the $\delta$-location set is introduced, as proposed in \cite{xiao2015protecting}. This set represents the collection of locations where the user is most likely to appear at time $t$, and we denote it as $\varDelta \chi _t$. Specifically, the $\delta$-location set $\varDelta \chi _t$ is defined as a set containing the minimum number of locations at time $t$ with a prior probability sum not less than $1-\delta $, where $0<\delta <1$. Formally, it is defined as
\begin{equation}\label{eq:11}
\varDelta \chi _t=\min \left\{ loc_i|\sum_{loc_i}{\mathbf{p}_{t}^{-}\left[ i \right]}\ge 1-\delta \right\}.
\end{equation}
Since the $\delta$-location set includes locations with a high probability of user presence at time $t$, there exists a small probability that the user's real location $\boldsymbol{x}_t$ of the user be excluded. In such a case, we substitute the closest location $\boldsymbol{\tilde{x}}_t$ for the real location $\boldsymbol{x}_t$, given by
\begin{equation}
\boldsymbol{\tilde{x}}_t=\underset{\boldsymbol{\tilde{x}}_t\in \varDelta \boldsymbol{\chi }_t}{\text{arg}\min}\ d\left( \boldsymbol{\tilde{x}}_t,\boldsymbol{x}_t \right).
\end{equation}
\begin{comment}
If $\boldsymbol{x}_t\in \varDelta \chi _t$, then $\boldsymbol{x}_t$ is protected in $\varDelta \chi _t$; if not, $\boldsymbol{\tilde{x}}_t$ is protected in $\varDelta \chi _t$.
\end{comment}
In this framework, it is considered protected if the user's real location $\boldsymbol{x}_t$ belongs to the $\delta$-location set (i.e., $\boldsymbol{x}_t\in \varDelta \chi _t$). Otherwise, the closest location $\boldsymbol{\tilde{x}}_t$ serves as the protected location to mitigate potential location inference risks.
\subsection{Concepts Related To Privacy Budget Allocation}
To more comprehensively and flexibly meet the personalized privacy protection needs of users at different times and locations, we consider multiple factors that influence the sensitivity of users' trajectories during the privacy budget allocation process. Based on this, we define several key concepts closely related to privacy budget allocation.
\begin{myDef}
\textnormal{\textbf{($\boldsymbol{w}$-Trajectory Sequence Differential Privacy)}} Consider a user's trajectory denoted as $Tr=\left\{ \boldsymbol{x}_1,\boldsymbol{x}_2,\boldsymbol{x}_3,\ldots,\boldsymbol{x}_T \right\} $. For any continuous sequence of $w$-timestamped location subsets
\begin{comment}
For any one of the $\mathrm{w}$ timestamped subsets of continuous location sequences
\end{comment}
 $\left\{ \boldsymbol{x}_{i-w+1},\boldsymbol{x}_{i-w+2},\ldots,\boldsymbol{x}_i \right\} $, if the cumulative privacy budget for this sequence satisfies (\ref{eq:13}), then the user's trajectory privacy is said to satisfy the $w$-trajectory sequence differential privacy, which is given by
\begin{equation}\label{eq:13}
\forall i\in[n],\sum_{k=i-w+1}^i\epsilon_k\leqslant\epsilon,
\end{equation}
where $\boldsymbol{x}_t$ represents the user's real location, which is perturbed to $\boldsymbol{x}_{t}^{\prime}$ by the privacy-preserving mechanism $\mathcal{M}_i$, which satisfies $\epsilon _i$-differential privacy.
\end{myDef}
 \begin{myDef}
\textnormal{\textbf{(Location Predictability)}} Location predictability $LP_{i,t}$ quantifies the predictability of the user's real location at time $t$, given by
\begin{equation}
LP_{i,t}=\frac{1}{1+\sum_{j=1}^n{\mathrm{p}_{t}^{-}\left( loc_i \right) \mathrm{p}_{t}^{-}\left( loc_j \right) d_3\left( loc_i,loc_j \right)}},
\end{equation}
where $\mathrm{p}_{t}^{-}\left( loc_i \right)$ and $\mathrm{p}_{t}^{-}\left( loc_j \right) $ denote the probabilities that, at time $t$, the user's real location is $loc_i$ and the attacker's inferred location is $loc_j$, respectively. $d_3\left( loc_i,loc_j \right)$ denotes the Euclidean distance between $loc_i$ and $loc_j$ in 3D space.
\end{myDef}
 \begin{myDef}
\textnormal{\textbf{(Location Sensitivity)}} Location sensitivity $LS_{i,t}$ represents the sensitivity of $loc_i$, taking into account the user's sojourn time, visit frequency of $loc_i$, and the semantic sensitivity of $loc_i$, given by
\begin{equation}
LS_{i,t}=\gamma _tT_i+\gamma _{_f}F_i+\gamma _{_s}Sen_i\cdot (1+\lambda I_{\mathrm{user}}),
\end{equation}
where $T_i$ and $F_i$ denote the user's sojourn time and visit frequency at $loc_i$, respectively. These values can be derived from the user's historical trajectory data.
\begin{comment}
 $T_i$ can be calculated through the user's mobile data in a certain period of time its access to different locations of the length of time, and
\begin{equation}
T_i=\frac{\int_{t_2}^{t_1}{f\left( t \right)}dt}{t_2-t_1}
\end{equation}
the value of $F_i$ indicate the frequency of the user's access to $x_i$.
\end{comment}
$Sen_i$ denotes the semantic sensitivity of $loc_i$ to the user, which can be configured based on the potential risk of privacy exposure with the semantic context of the location. Furthermore, individual privacy protection requirements can vary even within the same location, depending on user-specific attributes. For instance, the sensitivity attributed to a hospital location is typically lower for a physician compared to a patient. The characteristic adjustment coefficient, $I_{\text{user}}$, is introduced to account for such distinctions. For instance, in the scenario where a patient visits a hospital, $I_{\text{user}} = 0.8$; conversely, for a doctor visiting the same hospital, $I_{\text{user}} = 0.1 $. The parameter $\lambda$ represents the weight of $I_{\text{user}}$, determining the extent of its impact.

The coefficients
$\gamma _t$, $\gamma _f$, and $\gamma _s$ indicate the respective weights assigned to sojourn time, visit frequency, and location semantic sensitivity in the calculation of location sensitivity, which satisfy
\begin{equation}
 \gamma _t+\gamma _{_f}+\gamma _s=1.
 \end{equation}
This allows the users to adjust the weighting parameters according to their specific needs. For example, for a patient, although the sojourn time and visit frequency at the hospital may be short, the hospital itself could still reveal sensitive health information, leading to a higher semantic sensitivity for the patient. Consequently, lower values can be assigned to $\gamma_t$ and $\gamma_f$, while a higher value is assigned to $\gamma_s$.
\end{myDef}

\begin{figure*}[htbp]
\begin{comment}
	\centerline{\includegraphics[width=20cm]{GGT2.pdf}}
\end{comment}
\centerline{\includegraphics[width=17cm, keepaspectratio]{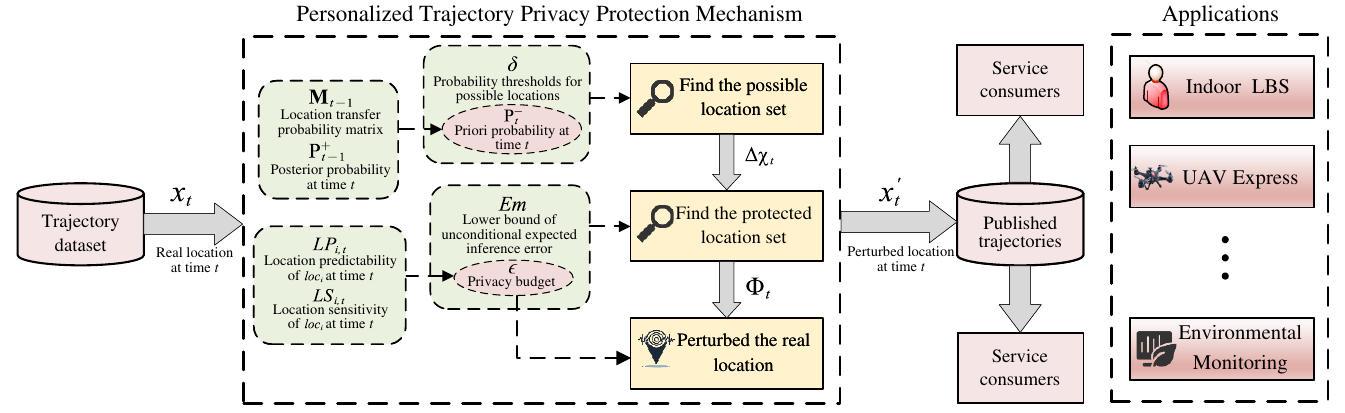}}
	\caption{The framework of 3DSTPM. 3DSTPM provides personalized trajectory privacy protection and is applicable to various scenarios. It operates in three stages: finding the possible location set, finding the protected location set, and perturbing the real location based on the privacy budget allocated according to the user's location sensitivity at different times.}
	\label{fig3}
\vspace{-0.2cm}
\end{figure*}
\subsection{Problem Statement}
  Compared to 2D space, incorporating height information in 3D space leads to a larger data scale and richer semantic content\cite{10325612}. This increases data processing complexity and poses challenges in establishing location sensitivity. In addition, unlike protecting singular locations' privacy, neglecting spatiotemporal correlations makes it easier for attackers to exploit this knowledge to infer users' sensitive information\cite{cao2024protecting}, heightening privacy breach risks. Thus, 3D spatiotemporal trajectory privacy protection must account for the dynamic changes in prior distributions across different trajectory locations due to spatiotemporal correlations. Directly applying existing mechanisms to 3D spatiotemporal trajectory privacy cannot handle 3D spatiotemporal trajectory data. Therefore, we focus on protecting the privacy of 3D spatiotemporal trajectory in this work.

 Significant spatiotemporal correlations often exist between locations on a trajectory, which attackers can exploit to enhance inference accuracy by updating their prior knowledge\cite{sym16091248}. For example, by utilizing a known user location at a given time and the transition probability matrix between locations, an attacker can infer the user's potential locations at subsequent time steps, posing a greater threat to user privacy. Furthermore, the sensitivity of different locations along a trajectory varies, and users' privacy protection needs change accordingly. Even for the same location, the sensitivity may fluctuate due to contextual changes. In addition, as continuous location-based queries are made along the trajectory, privacy accumulation issues may arise\cite{zhang2023personalized}. Continuous queries gradually expose the user's behavior patterns, significantly increasing the risk of privacy leakage.
\begin{comment}
Therefore, when protecting the privacy of user trajectories, we fully consider the spatio-temporal correlation between locations on the trajectory and dynamically update the prior probabilities of the user's different locations at different times. And based on the prior probabilities of different locations and the set probability threshold, we select the possible location set $\Delta \chi_t$ of the user, and protect all possible locations in $\Delta \chi_t$.
\end{comment}

To resolve this issue, we incorporate spatiotemporal correlation into our privacy protection strategy. By dynamically updating the prior probabilities of the user's potential locations at different times, we construct a possible location set $\Delta \chi_t$ for the user based on these probabilities and a predefined probability threshold. All locations within $\Delta \chi_t$ are subsequently protected. We leverage the complementary properties of 3D-GI and distortion privacy, employing a minimum distance search algorithm based on the 3D Hilbert curve to determine the protected location set $\Phi _t$ for each time step in the trajectory. Additionally, we adjust the privacy parameters $E_m$ and $\epsilon$ to cater to users' personalized privacy requirements, which is influenced by factors such as sojourn time, visit frequency, and location semantic sensitivity. To more flexibly meet users' personalized needs while addressing the privacy accumulation issue caused by continuous queries along the trajectory, we propose a W-APBA algorithm, which effectively considers location sensitivity attributes.
\begin{comment}
This method ensures that users' personalized privacy budget needs are met while maintaining compliance with $\mathrm{w}$-trajectory differential privacy.
\end{comment}
Moreover, to protect privacy while meeting users' QoS requirements, we apply the PF mechanism\cite{10325612} to perturb the real locations.
\begin{comment}
Finally, existing trajectory privacy protection mechanisms often prioritize privacy at the cost of service quality. To address this, we adopt the Permute-and-Flip (PF) mechanism to perturb real location, ensuring a balance between privacy protection and Qos. The PF mechanism aims to generate perturbed locations that are challenging to trace back to the real location. By minimizing the distance between the real and perturbed locations, the mechanism effectively balances privacy protection requirements with users' QoS demands.
\end{comment}
\section{Personalized Trajectory Privacy Protection Mechanism}\label{PERSONALIZED TRAJECTORY PRIVACY PROTECTION MECHANISM}
\vspace{-0.2cm}
\subsection{Overview}
\begin{algorithm}[t]
   \caption{Personalized 3D spatiotemporal trajectory privacy protection mechanism (3DSTPM)} \label{algorithm 1}
   \begin{small}
   \BlankLine
  \KwIn{$\mathbf{p}_{t-1}^{+}$, $\delta$, $E_m$, $Tr = [\boldsymbol{x}_1,\boldsymbol{x}_2,\ldots,\boldsymbol{x}_T]$, $\epsilon_s$, $\Delta \epsilon$, $\mathbf{M}$, and $w$}
   \KwOut{$\mathbf{p}_{t}^{+}$, $\boldsymbol{x}_{t}^{^{\prime}}$}
       $\mathbf{p}_{t}^{-}$=$\mathbf{p}_{t-1}^{+}$$\mathbf{M}$;\\
       $\epsilon_{i,t} = \mathbf{Algorithm 1}(\epsilon_s,\Delta \epsilon,\mathbf{M},w)$;\\
       Determine $\Delta \chi_t$ from $\mathbf{p}_{t}^{-}$, $\delta$;\\
   \If { $\boldsymbol{x}_t \notin \Delta \chi_t$ }{
        $\boldsymbol{x}_t = \tilde{x}_t$ via (11);}

    \For{$loc_i \in \Delta \chi_t, \; i = 1 \ \textbf{$\mathbf{to}$} \; \text{size of } \Delta \chi_t$  }{
        Find PLS $\Phi_t$ via Multiple rotated 3D Hilbert curve\cite{10325612};\\
        Release perturbed locations $\boldsymbol{x}_{t}^{^{\prime}}$ by PF mechanism;
        }
        Update $\mathbf{p}_{t}^{+}$ via (2);\\
        Go to next timestamp;\\
    \end{small}
\end{algorithm}
In this section, we propose a personalized 3D spatiotemporal trajectory privacy protection mechanism, named 3DSTPM, as illustrated in Fig. \ref{fig3}. This mechanism is designed to protect the trajectory privacy of objects moving in 3D space, which also remains applicable to 2D space by setting the height coordinate to zero. To build this mechanism, we consider the spatiotemporal correlations between different locations in a 3D spatiotemporal trajectory and provide personalized trajectory privacy protection for users by integrating 3D-GI with expected inference error.
\begin{comment}
we analyze the relevant characteristics of 3D spatio-temporal trajectories, fully consider the factors that affect the sensitivity of user trajectories, and determine the constraints of different factors on the measurement of trajectory sensitivity.
\end{comment}
\begin{comment}
Finally, we combine 3D-GI and distortion privacy to improve the robustness of the mechanism while meeting the individual needs of users.
\end{comment}
\begin{comment}
 Specifically, we first utilize the spatio-temporal correlations between locations at different timestamps within the trajectory to determine a possible location set for the user, based on the probability threshold and the prior probability of each location. Then, based on the user's privacy protection needs for the trajectory, we use a multiple rotated 3D Hilbert curve to find a protected location set with the smallest diameter for each possible location, thereby better balancing the needs of service quality while protecting privacy; in addition, we use the PF mechanism proposed in previous work \cite{cao2024protecting}, to generate a perturbed location $x_{t}^{'}$ for each location in PLS. What’s more, when we assign a uniform privacy budget to all locations on the trajectory, it will lead to insufficient protection of highly sensitive locations and excessive protection of low-sensitive locations. Therefore, we assign privacy budgets to different locations on the trajectory according to their location sensitivity to meet the personalized needs of users.
\end{comment}
Specifically, we first extract the location transition probability matrix from the user's historical trajectory dataset, and then, by combining the prior probability of each location with the probability threshold for possible locations, determine the possible location set. Next, by integrating the complementary characteristics of 3D-GI and distorted privacy, we identify a PLS for all locations within the possible location set, thereby enhancing the robustness of the mechanism while meeting users' personalized privacy needs. To enhance user QoS, we utilize a multiple-rotated 3D Hilbert curve to find the PLS with the smallest diameter. Additionally, we employ PF mechanism \cite{10325612} to generate a perturbed location $\boldsymbol{x}_{t}^{'}$ for each location in the PLS, achieving a balance between trajectory privacy protection and QoS. The process of the 3DSTPM is summarized in Algorithm \ref{algorithm 1}.
\begin{figure*}[htbp]
\vspace{0 cm}
	\centerline{\includegraphics[width=10cm]{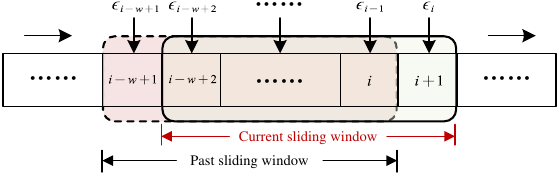}}
	\caption{$w$-sliding window.}
	\label{fig4}
\vspace{-0.4cm}
\end{figure*}

Furthermore, the sensitivity of different locations along a trajectory varies and may change dynamically over time, resulting in varying privacy protection requirements for users across different times and locations. In addition, continuous location-based queries along the trajectory lead to the gradual accumulation of privacy leakage, significantly increasing privacy risks. To overcome these challenges, we propose a novel privacy budget allocation algorithm, fully considering the sensitivity variations of locations over time and limiting the cumulative privacy leakage caused by continuous queries.
\subsection{Determine $\varDelta \chi_t$ at Continuous  Times}
\begin{comment}
The transition probability matrix $\mathbf{M}$ is constructed based on the user's historical trajectory data and behavior patterns\cite{xiao2015protecting}.
\end{comment}
To determine the possible location set $\varDelta \chi_t$ at time $t$, we first eliminate all impossible locations where $\mathbf{p}_{t}^{-}$ is minimal or equal to zero, following the criteria specified in (\ref{eq:11}). This ensures that frequently visited locations of users are effectively protected, thereby enhancing privacy protection. If the real location at time $t$ is removed, it is replaced with $\boldsymbol{\tilde{x}}_t$ to mitigate potential location inference risks.

Next, the posterior probability $\mathbf{p}_{t}^{+}$ is calculated according to (\ref{eq:3}), considering the spatiotemporal correlations between locations along the trajectory to enhance robustness against attackers who exploit such correlations. Using this posterior probability, the prior probability $\mathbf{p}_{t+1}^{-}$ for the next time step is then updated based on (\ref{eq:4}). From the updated $\mathbf{p}_{t+1}^{-}$ we derive $\varDelta \chi _{t+1}$. The size of $\varDelta \chi _{t+1}$ is determined by setting the value of $\delta$, which can be adjusted based on the user's privacy protection needs. This process is iteratively applied to obtain $\varDelta \chi _t$ for consecutive time steps.
\begin{comment}
We eliminate all impossible locations ($\mathbf{p}_{t}^{-}$ is minimal or $\mathbf{p}_{t}^{-}=0$) based on certain criteria to obtain the set of possible locations at time $t$, i.e., $\varDelta \chi _t$.
If the actual location at time $t$ is removed, we substitute it with $\boldsymbol{\tilde{x}}_t$.
We calculate the posterior probability $\mathbf{p}_{t}^{+}$ according to Eq. (\ref{eq:3}), and then combine the location transition probability matrix $\mathbf{M}$ according to Eq. (\ref{eq:4}): to obtain the prior probability $\mathbf{p}_{t+1}^{-}$ at time $t+1$. In terms of $\mathbf{p}_{t+1}^{-}$, we get $\varDelta \chi _{t+1}$ at time $t+1$. We determine the size of $\varDelta \chi _{t+1}$ by setting the value of $\delta$. Then, we obtain $\varDelta \chi _t$ at consecutive times by following the same process.
\end{comment}

\vspace{-0.2cm}
\subsection{Determine Protection Location Set}
After obtaining $\Delta \chi _t$ of the user at the current moment, we provide protection for all locations within $\Delta \chi _t$ by finding a PLS $\Phi _t$ for each possible location. 3D-GI can effectively limit attackers from further inferring users' personal information by observing perturbed locations released at different times. However, it cannot restrict the attacker's expected inference error. On the other hand, the distortion privacy can limit the attacker's inference accuracy, but its effectiveness depends heavily on the attacker’s prior knowledge. To leverage the advantages of both approaches, we integrate them to enhance the robustness against attackers with knowledge of spatiotemporal correlation information, constraining the attacker's inference error while independent of prior knowledge.
\begin{comment}
Therefore, we combine the two and utilize their complementary properties to improve the robustness of the mechanism by limiting the attacker's expected inference error while restricting the gain of relevant information.
\end{comment}

For any observed perturbed location, the conditional expectation error is formulated as\cite{10325612}:
\begin{equation}
E\left( \Phi _t \right) =\min_{\hat{\boldsymbol{x}}_t\in \mathcal{A}} \sum_{\boldsymbol{x}_t\in \Phi _t}{\frac{\mathrm{Pr}\left( \boldsymbol{x}_t \right)}{\sum_{\boldsymbol{y}_t\in \Phi _t}{\mathrm{Pr}\left( \boldsymbol{y}_t \right)}}d_3\left( \hat{\boldsymbol{x}}_t,\boldsymbol{x}_t \right) .}
\end{equation}
In the worst case, the attacker's inferred location is narrowed down to an arbitrary location within the PLS containing the user's real location, the specific formula is given by\cite{10325612}
\begin{equation}
E\left( \Phi _t \right) =\min_{\hat{\boldsymbol{x}}_t\in \mathcal{A}} \sum_{\boldsymbol{x}_t\in \Phi _t}{\frac{\mathrm{Pr}\left( \boldsymbol{x}_t \right)}{\sum_{\boldsymbol{y}_t\in \Phi _t}{\mathrm{Pr}\left( \boldsymbol{y}_t \right)}}d_3\left( \hat{\boldsymbol{x}}_t,\boldsymbol{x}_t \right) .}
\end{equation}
From (\ref{eq:7}), we obtain
\begin{equation}
ExpEr(\boldsymbol{x}_{t}^{\prime})\ge e^{-\epsilon}E(\Phi _t).
\end{equation}
Here, we introduce the privacy parameter $E_m$, which allows users to determine the unconditional expectation error lower term, i.e. $e^{-\epsilon}E(\Phi _t)\geqslant E_m$, according to their privacy preserving needs. To ensure this constraint is satisfied, we impose the following condition:
\begin{equation}
\setlength{\abovedisplayskip}{0.04cm}
\setlength{\belowdisplayskip}{0.01cm}
E(\Phi _t)\geqslant e^{\epsilon}E_m.
\end{equation}

The locations within the PLS $\Phi _t$ can satisfy the condition:
\begin{equation}
 ExpEr(\boldsymbol{x}_{t}^{\prime})\ge e^{-\epsilon}E(\Phi _t),
 \end{equation}
 indicating that a larger PLS generally provides stronger privacy protection. However, increasing the size of the PLS also results in greater QoS loss. To better balance trajectory privacy protection and QoS, we build upon the multi-rotation 3D Hilbert curves proposed in \cite{10325612}, and integrate the prior probability distribution that dynamically varies over time into the search process, so as to construct minimum-diameter PLS for all possible locations on the trajectory that better align with its spatiotemporal correlation characteristics. For the $\boldsymbol{x}_t$ in $\Delta \chi _t$, along the multi-rotation 3D Hilbert curve direction traverses all possible locations of the user and selects the sphere with the smallest diameter as the location protection set PLS $\Phi _t$.

Furthermore, let $D\left( \Phi _t \right)$ denote diameter within the PLS, which is defined as the maximum distance between any two locations within the PLS at time $t$, i.e.,
$D\left( \Phi _t \right) \geqslant d\left( \boldsymbol{x}_t,\hat{\boldsymbol{x}}_t \right)$. Under this definition, we derive the following bound:
\begin{equation}\label{eq:22}
e^{\epsilon}E_m\le E(\Phi _t)\le \min_{\hat{\boldsymbol{x}}_t\in \Phi _t} \sum_{\boldsymbol{x}_t\in \Phi _t}{\frac{\mathrm{Pr}(\boldsymbol{x}_t)}{\sum_{\boldsymbol{y}_t\in \Phi _t}{\mathrm{Pr} (\boldsymbol{y}_t)}}\mathrm{D(}\Phi _t)}=\mathrm{D(}\Phi _t).
\end{equation}
Thus, the diameter of the PLS must be no smaller than $e^{\epsilon}E_m$.

\vspace{-0.2cm}
\subsection{Privacy Budget Allocation Mechanism}
\begin{comment}
In order to meet the personalized needs of our users, we assign privacy budgets to all locations on the trajectory based on the varying sensitivity of the location and the geospatial constraints between it and other locations.
\end{comment}
To address the dynamic variation in location sensitivity so as to meet users’ personalized requirements and reduce cumulative privacy leakage arising from continuous location-based queries, we introduce the W-APBA algorithm, which dynamically allocates privacy budgets to different locations along the trajectory. The $w$-sliding window refers to the sequence of locations with timestamps of length $w$ that are continuously queried by the user, as illustrated in Fig. \ref{fig4}.

\begin{comment}
First, users sets the total privacy budget $\epsilon _s$ of the trajectory according to their privacy protection needs. Based on this, we assign the initial privacy budget $\epsilon _r^{t}=\frac{\epsilon _s}{w}$ to each location on the trajectory. However, different locations require different degrees of privacy protection due to their predictability and sensitivity, and therefore, we set the privacy increment $\varDelta \epsilon $ and the privacy budget control factor $\lambda _{i}^{t} $ to adjust the privacy budget, which is denoted by $\epsilon _{i}^{t}=\epsilon _r-\lambda _{i}^{t}\varDelta \epsilon $.
The degree of privacy budget adjustment can be controlled by $\lambda _{i}^{t} $, which is determined by the predictability and sensitivity of the location, we defined it as:
\begin{equation}
\lambda _{i}^{t}=\alpha _1PP_{i}^{t}+\alpha _2L_i
\end{equation}
\end{comment}
We aim to protect all complete trajectories that contain the user's possible locations. To achieve this, the user first sets the sliding window size $w$ and the total privacy budget $\epsilon _s$ for all possible locations within the sliding window, based on their privacy protection needs. To ensure the practical usability of the method, the value of $w$ should be smaller than the length of the trajectory. On this basis, we assign an initial privacy budget $\epsilon _r=\small{\frac{\epsilon _s}{w\cdot n}}$ to each possible location within $w$ window, where $n$ denotes the number of possible locations of the user for each timestamp. Additionally, we define $\epsilon _{w}=\frac{\mathrm{\epsilon}_{\mathrm{s}}}{n}$ as the total privacy budget allocated to $w$ sliding window for a single specific trajectory.

The attacker gradually grasps the spatiotemporal correlation between different locations in the trajectory by continuously accumulating data on the user's behavioral patterns, travel habits, and other contextual information. This iterative update enhances the attacker's prior knowledge, leading to dynamic variations in the predictability of different locations over time. Moreover, the location sensitive fluctuates depending on temporal and environmental contexts, resulting in users' varying privacy protection requirements for different locations within the trajectory at different moments.

To accommodate these dynamically changing privacy requirements, we introduce the privacy increment $\Delta \epsilon $ and the privacy budget control coefficient $\lambda$, which enable the system to dynamically allocate privacy budgets to potential locations within the trajectory at different moments. This approach provides more precise control over privacy protection strength. Based on this framework, we define the privacy budget for location $i$ at time $t$ as
\begin{equation}
\epsilon _{i,t}=\epsilon _r-\lambda _{i,t}\Delta \epsilon,
\end{equation}
where $\Delta \epsilon =\small{\frac{\epsilon_s}{2}}$, and $\lambda _{i,t} $ is determined by location predictability and location sensitivity, and we define $\lambda _{i,t} $ as
\begin{equation}
\lambda _{i,t}=\alpha _1LP_{i,t}+\alpha _2LS_{i,t}.
\end{equation}

Here, $LP_{i,t}$ is the predictability of the user's real location $\boldsymbol{x}_t$, with higher predictability requiring stronger privacy protection. However, predictability alone does not fully capture user's privacy preferences. Even when two locations exhibit the same predictability, users may have different privacy needs. To better accommodate users' personalized preferences, we also incorporate location sensitivity $LS_{i,t}$, which quantifies the sensitivity of $loc_i$. The coefficients $\alpha _1$ and $\alpha _2$ represent the weights of $LP_{i,t}$ and $LS_{i,t}$, respectively, and can be adjusted by users based on their individual privacy preferences.
\begin{comment}
In addition, due to the limitation of the sliding window, the privacy budget of the current location will be affected by the privacy budget of the previous location. The total privacy budget that can be allocated in the current window can be obtained by calculating the privacy budget consumption of the previous $\mathrm{w-1}$ locations according to Eq.(\ref{eq:20}). which is also the maximum value of the privacy budget allocated for the current location.
\end{comment}

To further address the cumulative privacy leakage caused by continuous queries along the trajectory, the sliding window mechanism dynamically limits the total privacy budget that can be allocated at a given time. Specifically, the maximum privacy budget available for the current location is defined as
\begin{equation}\label{eq:20}
\vspace{-0.1cm}
\epsilon _{\max ,t}=\epsilon _{w}-\sum_{k=i-w+1}^{i-1}{\epsilon _k},
\end{equation}
where $\epsilon _k$ represents the privacy budget allocated to the user's real location at time $k$.
\begin{algorithm}[t]
\vspace{0 cm}
   \caption{Window-based Adaptive Privacy Budget Allocation (W-APBA)} \label{algorithm 2}
   \begin{small}
   \BlankLine
   \KwIn{$\epsilon_s$, $\Delta \epsilon$, $\mathbf{M}$, and $w$}
   \KwOut{privacy budget: $\epsilon _{i,t}$}
   \For{location $loc_i$, $i = 1$ \textbf{$\mathbf{to}$} $n$ }{
         $\mathbf{p}_{t}^{-}=\mathbf{p}_{t-1}^{+}\mathbf{M}$; \\
    \For{\textit{Inferred locaion} $x_j$, $j = 1$ \textbf{$\mathbf{to}$} $n$ }{
        Calculate $d_3(loc_i, loc_j)$; \\
        Location predictability $LP_{i}^{t}$ via (13);\\
        Location Sensitivity $LS_{i}^t$ via (14);\\
        $\lambda _{i,t} = \alpha_1LP_{i}^{t}+\alpha_2LS_{i}^{t}$;\\
        Calculate maximum privacy budget $\epsilon _{\max,t}$ via (22);\\%=\epsilon _s-\sum_{k=i-w+1}^{i-1}{\epsilon_k}$;\\
        $\epsilon _{t} = \min(\epsilon_{t},\epsilon _{\max,t})$;\\
        $\epsilon _r = \frac{\epsilon _t}{n}$;\\
        $\epsilon _{i,t}= \epsilon _r-\lambda _{i,t}\Delta \epsilon$;
    }
}
RETURN $\epsilon _{i,t}$
   \end{small}
   \end{algorithm}

Given the constraint of $\epsilon _{\max,t}$, the privacy budget assigned to the current location is
\begin{equation}
\epsilon _{i,t}=\min\mathrm{(}\epsilon _r-\lambda _{i,t}\Delta \epsilon ,\epsilon _{\max,t}).
\end{equation}

W-APBA not only ensures that the trajectory satisfies the $w$-trajectory sequence differential privacy, but also allocates different privacy budgets to different locations based on location predictability and location sensitivity, which fully meets the personalized needs of users.

As illustrated in Algorithm \ref{algorithm 2},  the overall process of the proposed W-APBA algorithm is summarized. The computational complexity of the algorithm is $O(n^3)$, determined by the number of locations $n$ in the user's trajectory. This relatively low complexity enables the algorithm to perform privacy protection computations efficiently on large-scale trajectory datasets, thereby reducing computational overhead and improving processing efficiency.
\subsection{Differentially Private Mechanism in Protection Location Set}
Based on the allocated privacy budget, we apply the PF mechanism proposed in \cite{10325612} to perturb all the real locations along the trajectory. Furthermore, we demonstrate that the distance between the real and perturbed locations generated by the PF mechanism is constrained by a controllable upper bound, indicated in the following Theorem. Additionally, the algorithmic complexity of PF is $O(n^2)$, influenced by the number of possible user locations at different times along the trajectory\cite{10325612}. This relatively low complexity enables PF to efficiently process large-scale trajectory data while maintaining privacy protection.
\begin{Theorem}
The distance between the user's real location $\boldsymbol{x}_t$ and the perturbed location $\boldsymbol{x}_t^{\prime}$ sampled from the possible location set $\varDelta \chi _t$ by the PF mechanism, with probability at least $1-\psi, (0 \leq \psi \leq 1)$  satisfies the following inequality:
\begin{equation}\label{eq:27}
\setlength{\abovedisplayskip}{0.1cm}
\setlength{\belowdisplayskip}{0.1cm}
\begin{aligned}
d_3(\boldsymbol{x}_t, \boldsymbol{x}_t^{\prime}) \le & \frac{2D(\Phi_t)}{\epsilon} \Bigg[ \ln |\Delta \chi_t| - \frac{\epsilon}{2} - \ln |\Phi_t| - \ln \psi \\
& - \max d_3(\boldsymbol{x}_t, \boldsymbol{x}_t^{\prime}) \Bigg]\hspace{15pt} ( 0 \leq \psi \leq 1).
\end{aligned}
\end{equation}
\end{Theorem}
\begin{proof}
\vspace{-0.2cm}
The probability of a location being selected from $\varDelta \chi _t$ is given by\cite{cao2024protecting}
\begin{equation}\label{eq:28}
w_{\boldsymbol{x}}\exp \bigl( \frac{-\epsilon \left( d_3\left( \boldsymbol{x}_t,{\boldsymbol{x}_t}^{\prime} \right) -\max d_3\left( \boldsymbol{x}_t,{\boldsymbol{x}_t}^{\prime} \right) \right)}{2D(\Phi _t)} \bigr),
\end{equation}
where $w_{\boldsymbol{x}}$ is the probability distribution normalization factor.
Based on (\ref{eq:28}), we derive an upper bound on the total probability of the output perturbed locations that satisfy $d_3\left( \boldsymbol{x}_t,{\boldsymbol{x}_t}^{\prime} \right) \geqslant a,\left( a\geqslant 0 \right) $ within $\varDelta \chi _t$, i.e.,

\begin{small}
\begin{equation}
 w_{\boldsymbol{x}}|\varDelta \chi _t|\exp \bigl( \frac{-\epsilon \left( a-\max d_3\left( \boldsymbol{x}_t,{\boldsymbol{x}_t}^{\prime} \right) \right)}{2D(\Phi _t)} \bigr).
\end{equation}
\end{small}

According to\cite{10325612}, for any $\boldsymbol{x}_t^{\prime}$, $w_{\boldsymbol{x}}\leqslant \frac{e^{\frac{-\epsilon}{2}}}{\left| \Phi _t \right|}$. Thus, the above equation becomes
\begin{small}
\begin{equation}
\le \frac{|\varDelta \chi _t|e^{-\frac{\epsilon}{2}}}{|\Phi _t|}\exp \left( \frac{-\epsilon a-maxd_3\left( \boldsymbol{x}_t,{\boldsymbol{x}_t}^{\prime} \right)}{2D(\Phi )} \right).
\end{equation}
\end{small}
Let the right-hand side be $\psi$, we have:
\begin{small}
\begin{equation}
a=\frac{2D(\Phi )}{\epsilon}\left[ \left( \ln |\varDelta \chi _t|-\ln \psi -\ln |\Phi _t|-\frac{\epsilon}{2} \right) -\max d_3\left( \boldsymbol{x}_t,{\boldsymbol{x}_t}^{\prime} \right) \right].
\end{equation}
\end{small}
Therefore, for any $\boldsymbol{x}_t^{\prime}$ sampled from $\varDelta \chi _t$ by the PF mechanism, with probability at least $1-\psi$, we have (\ref{eq:27}).
 \vspace{-0.2cm}
 \end{proof}
 \begin{figure}[t]
    \centerline{\includegraphics[width=0.35\textwidth]{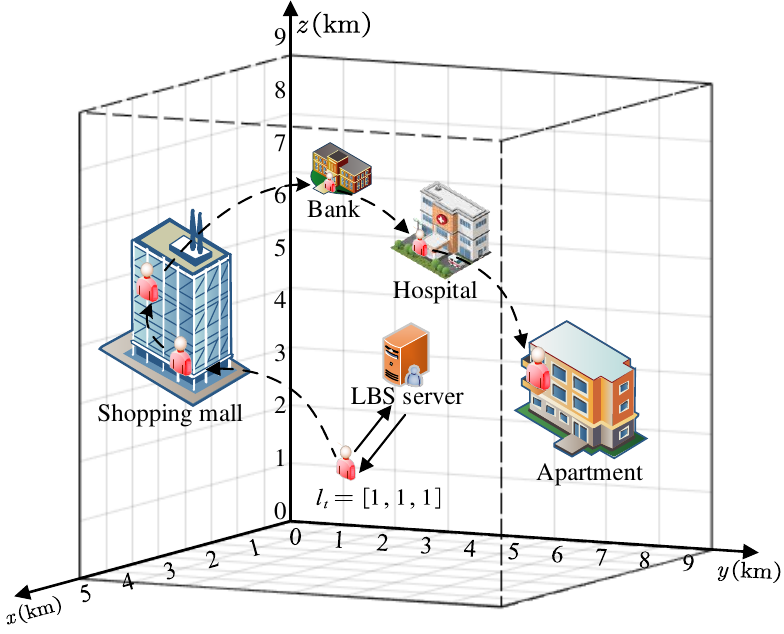}}
    \caption{Simulation setting of the trajectory of a user.}
    \label{fig5}
    \vspace{-0.1cm}
\end{figure}

\begin{figure*}[hbpt]
\vspace{0cm}                  %调整图片与上文距离
\setlength{\abovecaptionskip}{0cm}  %调整图片标题与图距离
\setlength{\belowcaptionskip}{0cm} %调整图片标题与下文距离
\centering
\subfigure[Trajectory privacy v.s. $\epsilon _w$ under same $E_m$]{
\includegraphics[height=4.2cm,width=5.8cm]{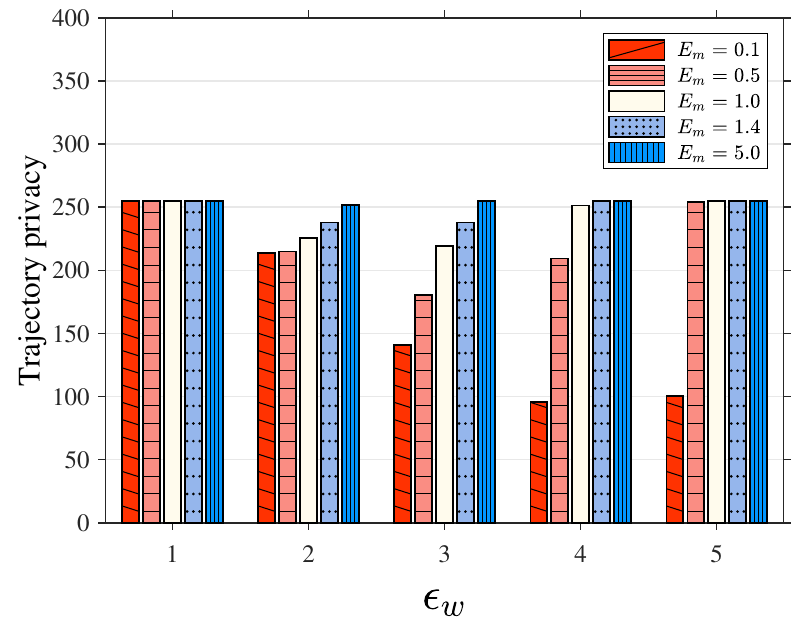} \label{fig6a}
}
\quad \quad \quad
\subfigure[QoS loss v.s. $\epsilon _w$ under same $E_m$]{
\includegraphics[height=4.2cm,width=5.8cm]{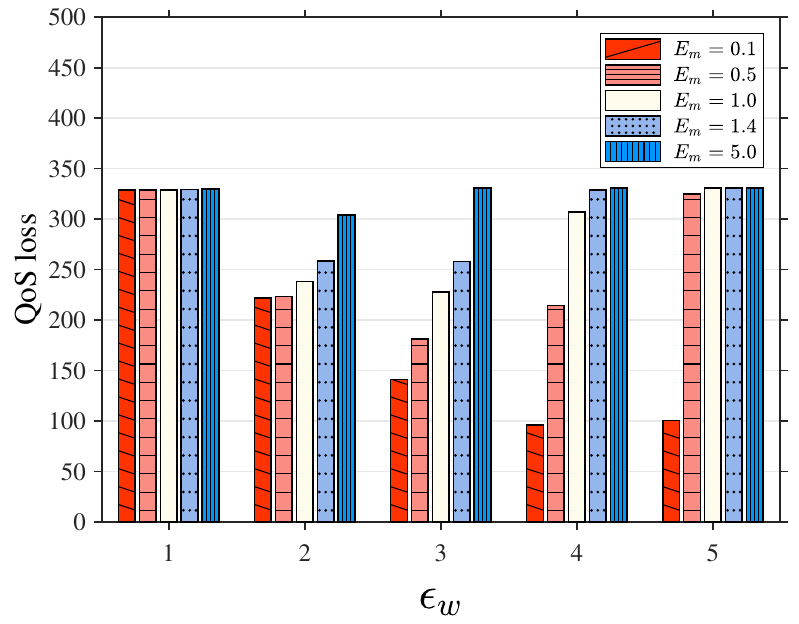} \label{fig6b}
}
\quad \quad \quad
\subfigure[Trajectory privacy v.s. $E_m$ under same $\epsilon _w$]{
\includegraphics[height=4.2cm,width=5.8cm]{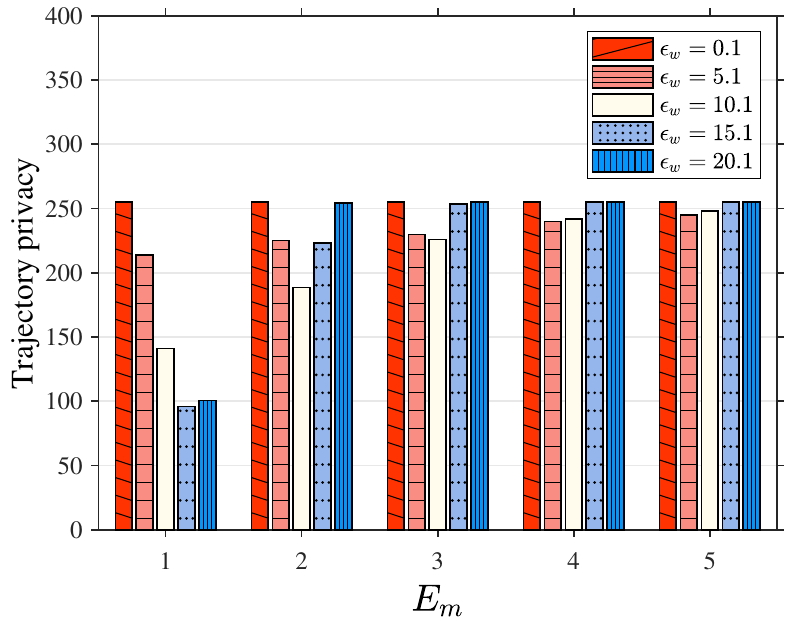} \label{fig6c}
}
\quad \quad \quad
\subfigure[QoS loss v.s. $E_m$ under same $\epsilon _w$]{
\includegraphics[height=4.2cm,width=5.8cm]{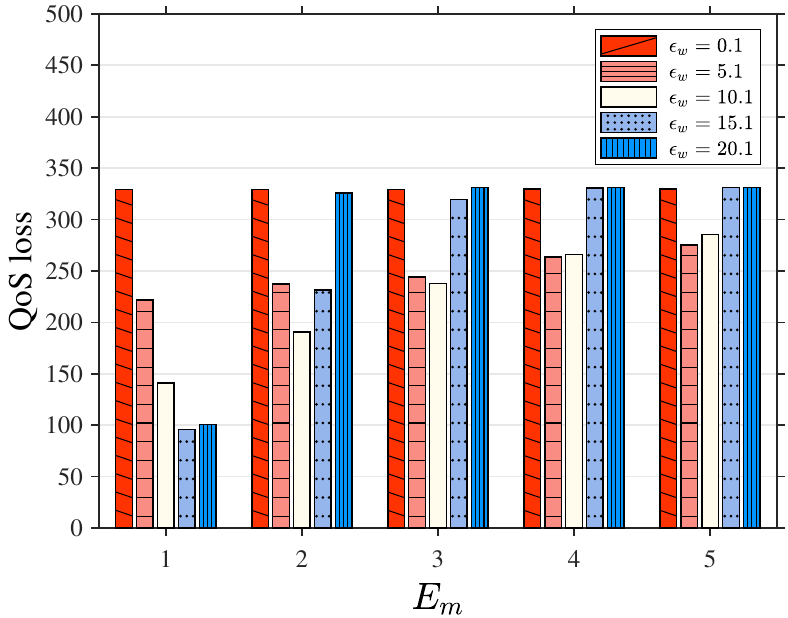} \label{fig6d}
}
\quad
\caption{Impact of $\epsilon _w$ and $E_m$ on personalized trajectory privacy protection.}
\label{fig6}
\end{figure*}
 \begin{remark}
The PF mechanism achieves a great balance between trajectory privacy and QoS by bounding the maximum perturbation distance from the real location. As demonstrated in (\ref{eq:27}), this upper bound is jointly determined by $\epsilon$ and $E_m$, which affect the diameter of the PLS. Therefore, by adjusting the values of $\epsilon$ and $E_m$, the mechanism exhibits flexibly to users’ personalized privacy needs while maintaining QoS.
\end{remark}
\section{Simulations Results}\label{SIMULATION RESULTS}
\vspace{-0.2cm}
\subsection{Simulation Setup}
\begin{comment}
In this section, we assess the effectiveness of our proposed P3DTPPM mechanism by comparing its performance with 3DPIVE and PTPPM in terms of the trajectory privacy and QoS loss. For evaluation, we divide a 10 km × 10 km × 10 km space into 512 equal units, representing the 3D space where the mobile users exist. Each unit corresponds to an area that users may access, known as the attacker's prior distribution. Each unit is associated with a specific location characterized by 3D coordinates. We select five of these units to represent the real locations at five consecutive moments along a user's trajectory, as depicted in Fig. 4.
\end{comment}
\begin{comment}
In this section, we assess the performance of the proposed 3DPTPPM and compare it to a mechanism that disregards spatio-temporal correlations 3DLPPM, as well as the 2DPTPPM, when facing an attacker who has knowledge of spatio-temporal correlations.Simulation results show that 3DPTPPM effectively improves the privacy protection level of 3D trajectories and reduces the privacy leakage of users.
\end{comment}
In this section, we analyze the impact of privacy parameters on the performance of our proposed mechanism 3DSTPM. Additionally, we compare the trajectory privacy and QoS loss performance of 3DSTPM with P3DLPPM\cite{10325612} and 2DPTPPM\cite{cao2024protecting} under attacks from attackers possessing spatiotemporal correlation knowledge. Moreover, the PIM  mechanism in\cite{xiao2015protecting} is also modified and applied in the 3D space, named 3DPIM, as a benchmark.

For the simulation, we divide a 10m × 10m × 10m 3D space into 512 cells, with each cell representing a potential 3D location of the user. We then randomly select 50 of these cells as the user's real locations, and the probability of the user reaching each of these locations is treated as the attacker's prior distribution. This setup simulates the attacker's initial knowledge of the user's location. Fig. \ref{fig5} depicts the trajectory of the user at five sequential moments.

The trajectory privacy $p$ and QoS loss $q$ are evaluated according to the metrics described in our previous work\cite{10325612}, which are given by
\begin{equation}
p=\sum_{\boldsymbol{x}_t,\boldsymbol{x}_{t}^{\prime},\hat{\boldsymbol{x}}_t\in \mathcal{A}}{\mathrm{Pr}\left( \boldsymbol{x}_t \right) f\left( \boldsymbol{x}_{t}^{\prime}\mid \boldsymbol{x}_t \right) h\left( \hat{\boldsymbol{x}}_t\mid \boldsymbol{x}_{t}^{\prime} \right) d\left( \boldsymbol{x}_t,\hat{\boldsymbol{x}}_t \right) ,}
\end{equation}
\begin{equation}
q=\sum_{\boldsymbol{x}_t,\boldsymbol{x}_{t}^{\prime}\in \mathcal{A}}{\mathrm{Pr}\left( \boldsymbol{x}_t \right) f\left( \boldsymbol{x}_{t}^{\prime}\mid \boldsymbol{x}_t \right) d\left( \boldsymbol{x}_t,\boldsymbol{x}_{t}^{\prime} \right) .}
\end{equation}

\subsection{Impact of Privacy Parameters on Personalized Trajectory Privacy Protection}
By setting different privacy budgets $\epsilon$ and inference error thresholds $E_m$, we evaluate the impact of various privacy parameters on the privacy protection performance of 3DSTPM. As illustrated in Fig. \ref{fig6}, it can be seen that the two privacy parameters $\epsilon$, $E_m$ have a significant impact on both trajectory privacy and QoS loss.
\begin{figure*}[t]
\centering
\subfigure[Trajectory privacy v.s. $w$ under same $\epsilon _w$]{
\includegraphics[height=4.2cm,width=5.8cm]{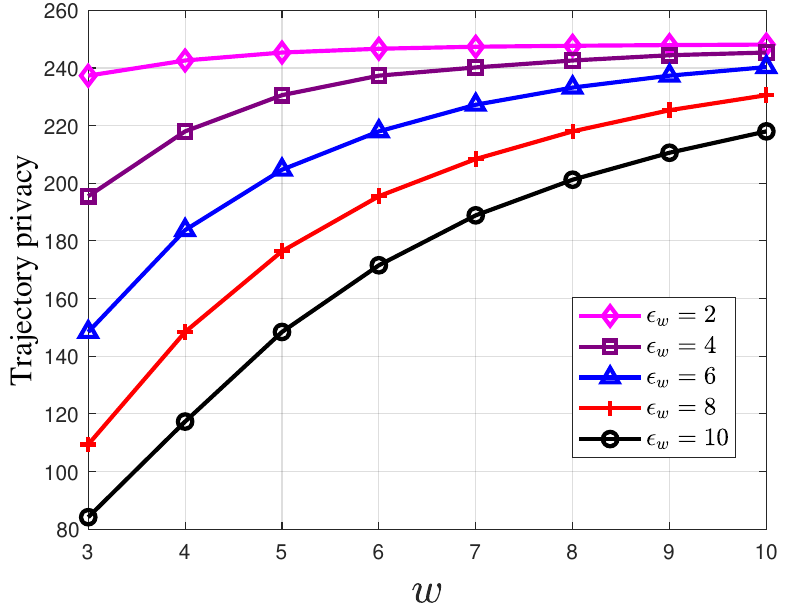} \label{fig7a}
}
\quad \quad \quad
\subfigure[QoS loss v.s. $w$ under same $\epsilon _w$]{
\includegraphics[height=4.2cm,width=6cm]{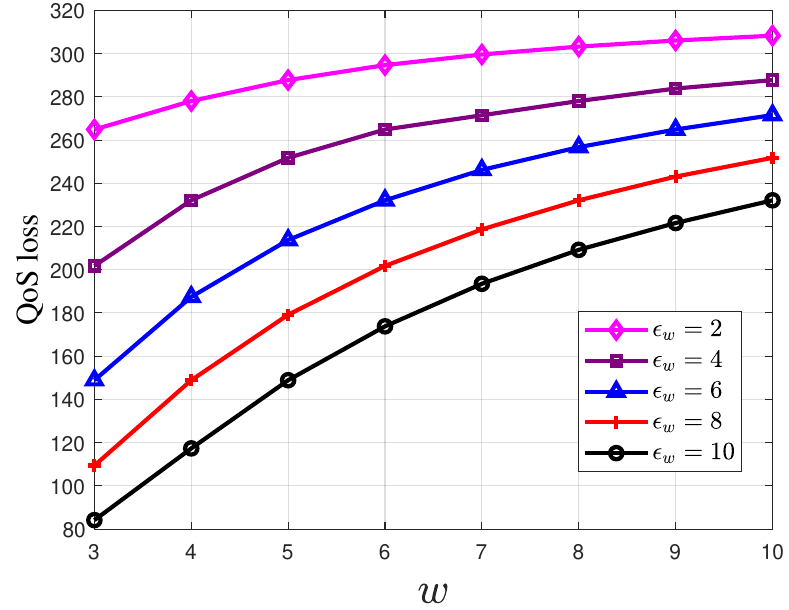} \label{fig7b}
}
\quad \quad \quad
\subfigure[Trajectory privacy v.s. $w$ under same $w$]{
\includegraphics[height=4.2cm,width=5.8cm]{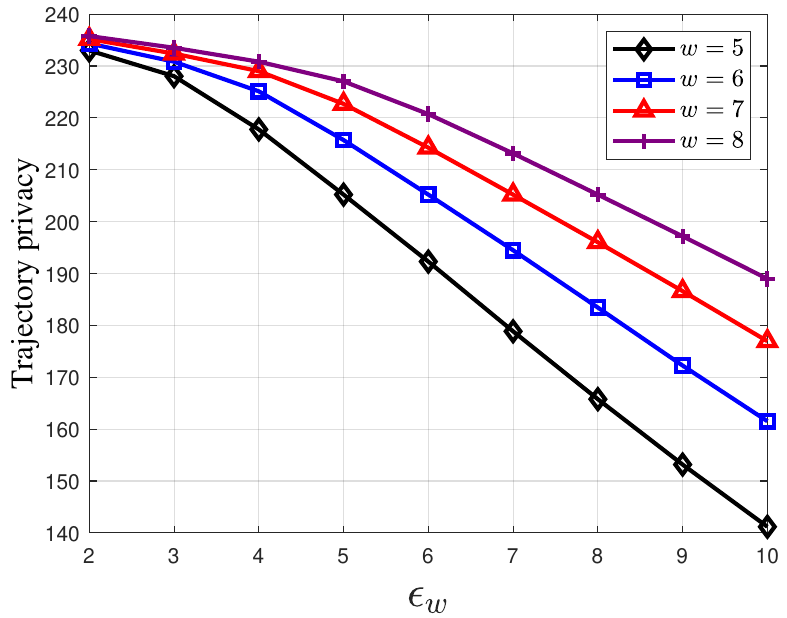} \label{fig7c}
}
\quad \quad \quad
\subfigure[QoS loss v.s. $w$ under same $w$]{
\includegraphics[height=4.2cm,width=5.8cm]{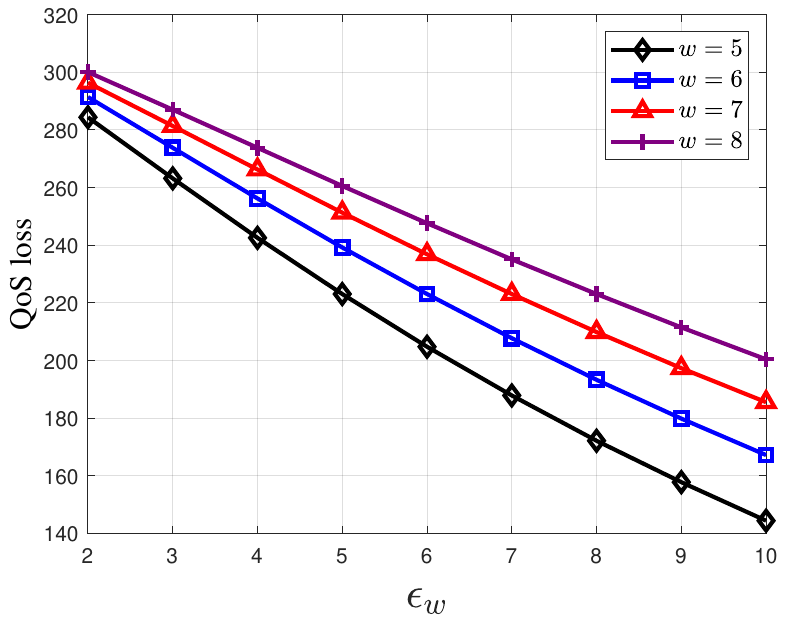} \label{fig7d}
}
\quad
\caption{Impact of W-APBA parameters on personalized trajectory privacy protection.}
\label{fig7}
\end{figure*}
\begin{figure}[t]
\begin{center}
\includegraphics[height=4.6cm,width=5.8cm]{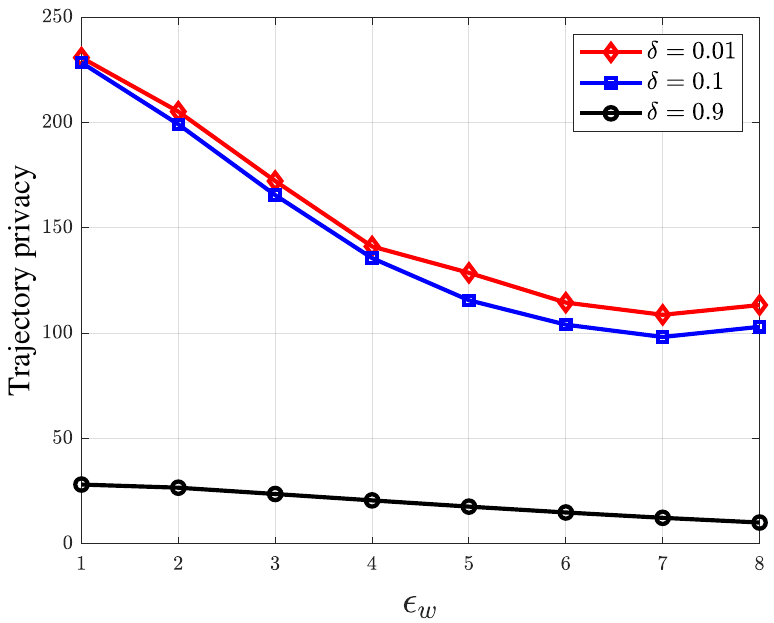}
\end{center}
\caption{Impact of probability thresholds $\delta$ on personalized trajectory privacy protection.}
\label{fig8}
\end{figure}

Specifically, from $e^{\epsilon}E_m\le \mathrm{D(}\Phi _t)$, it can be seen that the impact of $\epsilon$ on $\mathrm{D}\left( \Phi _t \right) $ exhibits exponential growth, while the impact of $E_m$ on $\mathrm{D}\left( \Phi _t \right) $ is linear. As shown in Figs. \ref{fig6a} and \ref{fig6b}, when $E_m$ is fixed and $\epsilon_w$ is small, increasing $\epsilon_w$ leads to a decrease in ${{D(\Phi _t)}/{\epsilon _w}}$,
 which, in turn, reduces the distance between the real location and the perturbed location. Consequently, both trajectory privacy and QoS loss gradually decrease. Moreover, as the value of $E_m$ increases, the trajectory privacy and QoS loss also increase. With a further increase in $\epsilon_w$, the diameter of the protection region increases sharply, causing ${{D(\Phi _t)}/{\epsilon _w}}$ to increase, which results in a larger distance between the real and perturbed locations. This effect varies depending on the settings of $E_m$. Nonetheless, in practical scenarios, the diameter of the protection region cannot increase indefinitely, so both trajectory privacy and QoS loss will eventually converge to an upper limit.

Figs. \ref{fig6c} and \ref{fig6d} illustrate the impact of $E_m$ on trajectory privacy and QoS loss under different values of $\epsilon_w$. Given a specific $\epsilon_w$, as $E_m$ increases, the diameter of the protection region increases, leading to higher trajectory privacy and QoS loss. However, when $\epsilon_w$ is sufficiently large, its effect on $\mathrm{D}\left( \Phi _t \right) $ is far greater than that of $E_m$. Thus, when $\epsilon_w=2.5$, an increase in $E_m$ causes a significant change in the diameter of the protection region, the distance between the user's real location and the perturbed location increases, resulting in a steep increase in trajectory privacy and QoS loss, making the curve steep. In practical applications, however, the diameter of the protection region is limited, so trajectory privacy and QoS loss will converge to a finite value. Additionally, an excessively small privacy budget allocation is insufficient to ensure an acceptable level of QoS loss. Therefore, regardless of the $E_m$ setting, when $\epsilon_w=0.1$, trajectory privacy and QoS loss reach their maximum values.

In Fig. \ref{fig7}, we explore the impact of the window size $w$ on trajectory privacy and QoS loss. As shown in Figs. \ref{fig7a} and \ref{fig7b}, when $\epsilon_w$ is fixed, the privacy budget allocated to each window decreases as $w$ increases. This leads to an increase in ${{D(\Phi _t)}/{\epsilon _w}}$, as shown in (\ref{eq:27}) in Theorem 1, the distance between the user's real location and the perturbed location increases, leading to an increase in trajectory privacy. For instance, when $\epsilon_w=8$, the trajectory privacy protection value is 176.41 with $w=5$, representing an 18.87\% increase compared to the value at $w=4$. Furthermore, the increase in the distance between the real location and the perturbed location also leads to a higher QoS loss. For example, the QoS loss value when $w=5$ is 20.34\% higher than the QoS loss value at $w=4$. Similarly, Figs. \ref{fig7c} and \ref{fig7d} demonstrate that when $w$ is fixed, the privacy budget allocated to each window increases as $\epsilon_w$ increases. This leads to a larger range of perturbed locations, gradually increasing the distance between the real location and the perturbed location, thereby reducing trajectory privacy while decreasing QoS loss.

Fig. \ref{fig8} illustrates the impact of the probability threshold $\delta$ on trajectory privacy $p$. As depicted in Fig. \ref{fig8}, trajectory privacy $q$ decreases as $\delta$ increases. When $\delta$ continues to rise towards 1, trajectory privacy experiences a significant decline, and the user's privacy can no longer be effectively protected. This occurs because, as $\delta$ increases, the number of possible user locations decreases, boosting the success rate of the attacker's inference. Consequently, trajectory privacy diminishes. When $\delta$ nears 1, the number of possible user locations can drop to just one, causing the attacker's inference success rate to soar, and the user's trajectory privacy is almost entirely lost, making it impossible to ensure adequate protection.
\subsection{Performance of Different Trajectory Privacy Protection Mechanisms}
Next, we quantitatively compare 3DSTPM with P3DLPPM, 3DPIM and 2DPTPPM in terms of trajectory privacy and QoS loss to verify its advantages.
\begin{comment}
For 2DPTPPM, we ignore height information while keeping all other information unchanged, focusing only on protecting the user’s planar trajectory information and calculating trajectory privacy and QoS loss based on this data.
\end{comment}
\begin{comment}
As shown in Fig.\ref{fig9}, in a 2D space without considering trajectory height information, locations with the same height along the trajectory may fall within restricted areas with the same semantic attributes. Regardless of how perturbed locations are released, this configuration may still reveal the semantic information of the user's location. For example, when $\epsilon_w=4.0$, the trajectory privacy value for P3DTPPM, which considers spatio-temporal correlations, is 195.4980, while for P3DLPPM, which does not consider spatio-temporal correlations, the trajectory privacy value is 155.6059. The trajectory privacy value of P3DTPPM consistently remains higher than that of P3DLPPM when spatio-temporal correlations are considered. In contrast, 2DPTPPM, which does not account for trajectory height information, has a trajectory privacy value of 86.7701, approximately 2.3 times lower.
\end{comment}
\begin{figure}[t]
\centering
\subfigure[Comparison of different mechanisms performance on Trajectory privacy.]{
\includegraphics[height=4.2cm,width=5.8cm]{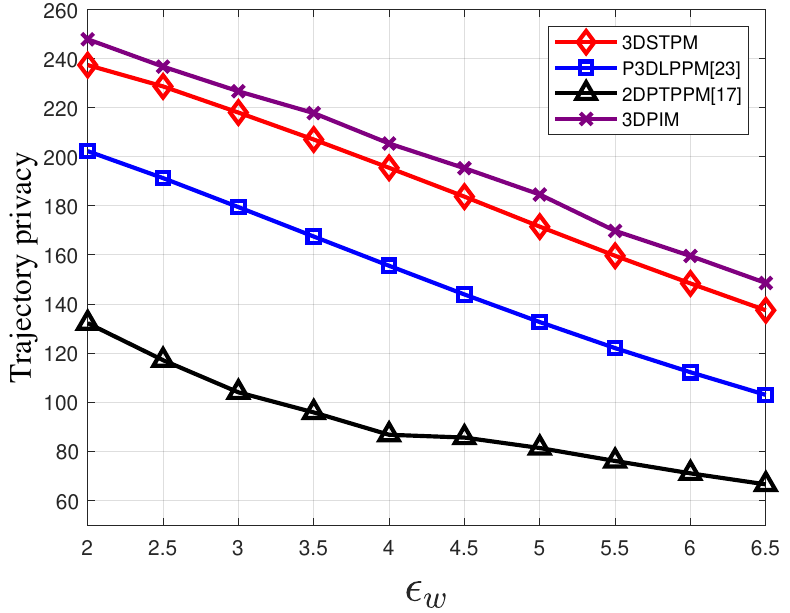} \label{fig9a}
}
\quad \quad \quad
\subfigure[Comparison of different mechanisms performance on QoS loss.]{
\includegraphics[height=4.2cm,width=5.8cm]{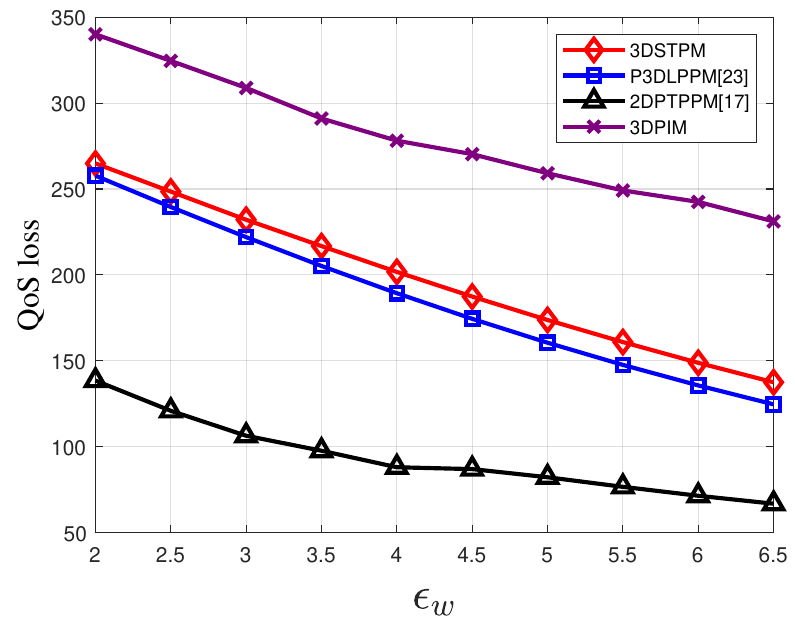} \label{fig9b}
}
\quad
\caption{Comparison of the performance of different mechanisms.}
\label{fig9}
\end{figure}

As illustrated in Fig. \ref{fig9}, 3DSTPM exhibits significant advantages over 2DPTPPM. When facing the same attacker with spatiotemporal correlation knowledge, 2DPTPPM demonstrates lower QoS loss but significantly falls short in trajectory privacy value compared to 3DSTPM, making it inadequate to meet the trajectory privacy protection requirements in 3D space. For instance, when $\epsilon_w=4.0$, the trajectory privacy value of 3DSTPM is approximately 2.3 times higher than 2DPTPPM. This disparity arises because 2DPTPPM perturbs only the horizontal coordinates of trajectory locations, overlooking height information, which can also pose a risk of privacy leakage. In contrast, 3DSTPM not only safeguards the horizontal trajectory but also protects the height data, which offers a higher level of privacy protection.

Fig. \ref{fig9} also shows that under similar QoS levels, 3DSTPM demonstrates a more significant advantage in trajectory privacy compared to P3DLPPM and 3DPIM. For the same $\epsilon_w$, the QoS loss of 3DSTPM is slightly higher than that of P3DLPPM, but the improvement in trajectory privacy significantly outweighs the increase in QoS loss. For instance, when $\epsilon_w=4.5$, the privacy level of 3DSTPM improves by 21.68\% compared to P3DLPPM, while the QoS loss increases by only 6.91\%. This advantage arises because P3DLPPM does not account for the spatiotemporal correlation between locations on the trajectory, instead offering protection solely for individual locations along the trajectory. In contrast, 3DSTPM enhances privacy protection by incorporating spatiotemporal correlations between trajectory locations and dynamically updating the prior probability of a user's location using the location transition probability matrix. Additionally, although the privacy level of 3DSTPM is slightly lower than that of 3DPIM, its QoS loss is significantly reduced. For instance, when $\epsilon_w=4.5$, the trajectory privacy of 3DSTPM is 6.30\% lower than that of 3DPIM, while the QoS loss is reduced by 30.66\%. The reason is that 3DSTPM employs the PF mechanism to perturb different locations along the trajectory, providing a smaller perturbation distance while guaranteeing privacy needs. As a result, it ensures trajectory privacy while minimizing QoS loss.

\section{Conclusion}\label{CONCLUSION}
\begin{comment}
In this paper, we propose a 3D personalized trajectory privacy protection mechanism. This paper has three novel contributions. First, we consider the privacy disclosure risk that may be caused by the dimension of height, and extend trajectory privacy protection from 2D to 3D space by introducing height information. Second, we consider the spatio-temporal correlation between locations on the trajectory, thereby improving the accuracy and robustness of trajectory privacy protection, effectively defending against powerful attackers with spatio-temporal correlation analysis capabilities. Finally, we fully consider the privacy attributes of location and propose a privacy budget allocation method based on the $w$ window, which adaptively adjusts the allocation of privacy budget to ensure that while protecting the privacy of highly sensitive locations, the overall efficiency of privacy budget usage is optimized. We combine the complementary properties of 3D-GI and expected error to adjust the privacy budget and the expected inference error threshold, improving the robustness of the mechanism while achieving personalized privacy protection of trajectories. Simulation results show that 3DPTPPM can effectively achieve personalized trajectory privacy protection, which can effectively improve the privacy protection effect compared with 2D trajectory privacy protection mechanisms, and has a significant performance advantage over 3DLPPM when attacked by an attacker with spatio-temporal correlation.
\end{comment}
In this paper, we have proposed a personalized 3D spatiotemporal trajectory privacy protection mechanism 3DSTPM. This paper has two novel contributions. First, we overcome the challenges posed by the high dimensionality of 3D space, consider the spatiotemporal correlations between locations on the trajectory, and combine the complementary characteristics of 3D-GI and expected error, expanding the applicability of the mechanism while enhancing its robustness. Second, we have proposed a W-APBA, which dynamically adjusts the privacy budget to provide more flexible privacy protection for users. We have also demonstrated that the distance between the real and perturbed locations generated by PF is constrained by a controllable upper bound, further reinforcing the theoretical foundation of the PF mechanism. Simulation results demonstrate that 3DSTPM significantly outperforms benchmarks in defending against spatiotemporally correlated attackers. For instance, when $\epsilon_w=5.0$, 3DSTPM achieves a privacy level that is 2.1 times higher than 2DPTPPM and 29\% higher than P3DLPPM, with only an 8\% increase in QoS loss compared to P3DLPPM. Moreover, although 3DSTPM has a 7.08\% lower privacy level than 3DPIM, its QoS loss is reduced by 32.94\%.

\begin{spacing}{1.0}
\balance
\bibliography{reference}% 指向同一个目录下的PU.bib 文件
\bibliographystyle{ieeetr}
\end{spacing}

\end{document}